\documentclass[10pt,final,letterpaper,twocolumns,romanappendices]{IEEEtran}
\usepackage{bm,cite}
\usepackage{amsmath}
\usepackage{amsthm}
\usepackage{amsbsy}
\usepackage{epsfig}
\usepackage{latexsym}
\usepackage{amssymb}
\usepackage{wasysym}
\usepackage{placeins} 

\newtheorem{theorem}{Theorem}
\newtheorem{definition}{Definition}
\newtheorem{proposition}{Proposition}

\newtheorem{corollary}{Corollary}

\newtheorem{example}{Example} %[section]

\DeclareMathOperator{\Dist}{dist}
\DeclareMathOperator{\Conv}{conv}

\DeclareMathAlphabet{\mathbit}{OML}{cmr}{bx}{it}
\DeclareMathAlphabet{\mathsf}{OT1}{cmss}{m}{n}
\DeclareMathAlphabet{\mathTXf}{OT1}{cmss}{bx}{it}

\DeclareMathOperator{\diag}{diag}

\DeclareMathOperator{\DoF}{DoF}

\DeclareMathOperator{\Rate}{R}

\DeclareMathOperator{\Size}{s}

\newcommand{\bB}{\mathbf{B}} 
 
\newcommand{\bD}{\mathbf{D}}

\newcommand{\bH}{\mathbf{H}}

\newcommand{\bS}{\mathbf{S}}

\newcommand{\LB}{\left(}
\newcommand{\RB}{\right)}
\newcommand{\LSB}{\left[}
\newcommand{\RSB}{\right]}

\newcommand*{\dotleq}{\mathrel{\dot{\leq}}}

\newcommand{\I}{\mathbf{I}} %identity matrix
\newcommand{\norm}[1]{\lVert{#1}\rVert}

\newcommand{\Fro}{{\mathrm{F}}}

\newcommand{\E}{{\mathrm{E}}}

\newcommand{\He}{{{\mathrm{H}}}}

\newcommand{\xv}{\mathbf{x}}

\theoremstyle{remark}
\newtheorem{remark}{Remark} %[section]
\theoremstyle{example}
\begin{document} 
\title{Spatial CSIT Allocation Policies for Network MIMO Channels}%
\author{Paul de Kerret and David Gesbert\\Mobile Communications Department, Eurecom\\
Campus SophiaTech, 450 Route des Chappes, 06410 Biot, France\\ {\normalsize \tt{\{{dekerret@eurecom.fr}{},{gesbert@eurecom.fr}{}\}}}} 
\maketitle 

\begin{abstract} %Furthermore, we formulate the problem of optimizing the CSIT allocation subject to a constraint on the total number of feedbacks bits in the network. 
In this work\footnote{This work has been performed under the Celtic-Plus project SHARING. Preliminary results have been published in \cite{dekerret2012_ICC}.}, we study the problem of the optimal dissemination of channel state information (CSI) among $K$ spatially distributed transmitters (TXs) jointly cooperating to serve $K$~receivers (RXs). One of the particularities of this work lies in the fact that the CSI is \emph{distributed} in the sense that each TX obtains its \emph{own} estimate of the global multi-user MIMO channel with no further exchange of information being allowed between the TXs. Although this is well suited to model the cooperation between non-colocated TXs, e.g., in cellular Coordinated Multipoint (CoMP) schemes, this type of setting has received little attention so far in the information theoretic society. We study in this work what are the CSI requirements at every TX, as a function of the network geometry, to ensure that the maximal number of degrees-of-freedom (DoF) is achieved, i.e., the same DoF as obtained under perfect CSI at all TXs. We advocate the use of the generalized DoF to take into account the geometry of the network in the analysis. Consistent with the intuition, the derived generalized DoF maximizing CSI allocation policy suggests that TX cooperation should be limited to a specific finite neighborhood around each TX. This is in sharp contrast with the conventional (uniform) CSI dissemination policy which induces CSI requirements that grow unbounded with the network size. The proposed CSI allocation policy suggests an alternative to clustering which overcomes fundamental limitations such as (i) edge interference and (ii) unbounded increase of the CSIT requirements with the cluster size. Finally, we show how finite neighborhood CSIT exchange translates into finite neighborhood message exchange so that finally global interference management is possible at finite SNR with only local cooperation. 
\end{abstract}

\IEEEpeerreviewmaketitle
%%%%%%%%%%%%%%%%%%%%%%%%%%%%%%%%%%%%%%%%%%%%%%%%%%%%%%%%%%%%%%%%%%%%%%%
%%%%%%%%%%%%%%%%%%%%%%%%%%%%%%%%%%%%%%%%%%%%%%%%%%%%%%%%%%%%%%%%%%%%%%%
%%%%%%%%%%%%%%%%%%%%%%%%%%%%%%%%%%%%%%%%%%%%%%%%%%%%%%%%%%%%%%%%%%%%%%%
%%%%%%%%%%%%%%%%%%%%%%%%%%%%%%%%%%%%%%%%%%%%%%%%%%%%%%%%%%%%%%%%%%%%%%%
\section{Introduction}
%%%%%%%%%%%%%%%%%%%%%%%%%%%%%%%%%%%%%%%%%%%%%%%%%%%%%%%%%%%%%%%%%%%%%%%
%%%%%%%%%%%%%%%%%%%%%%%%%%%%%%%%%%%%%%%%%%%%%%%%%%%%%%%%%%%%%%%%%%%%%%%
%%%%%%%%%%%%%%%%%%%%%%%%%%%%%%%%%%%%%%%%%%%%%%%%%%%%%%%%%%%%%%%%%%%%%%%
%%%%%%%%%%%%%%%%%%%%%%%%%%%%%%%%%%%%%%%%%%%%%%%%%%%%%%%%%%%%%%%%%%%%%%%

Network (or Multicell) MIMO methods, whereby multiple interfering transmitters (TXs) share user messages and allow for joint precoding, are currently considered for next generation wireless networks \cite{Gesbert2010}. With perfect message and channel state information (CSI) sharing, the different TXs can be seen as a unique virtual multiple-antenna array serving all receivers (RXs), in a multiple-antenna broadcast channel (BC) fashion. However, the sharing of the user's data symbols and the CSI to all cooperating TXs imposes huge requirements on the backhaul architecture, particularly as the number of cooperating TXs increases.

As a consequence, there has been a large literature dealing with the joint precoding across TXs based on limited backhaul (See \cite{Marsch2008,Simeone2009,Ng2008,Shamai2011,Park2013} and references therein). In \cite{Ng2008}, distributed schemes based on iterative updates of the transmit coefficients were designed to avoid the requirements of explicit CSI at the TXs. However, this approach cannot be applied in many delay-limited scenarios as the iterations introduce significant delay. In \cite{Simeone2009,Park2013}, considering centralized precoding, the impact of limited backhaul links between the central node and the TXs is discussed. In \cite{Marsch2008,Shamai2011}, the problem of the sharing of the users data symbols is studied with the assumption of perfect CSI at all TXs.

However, when considering the imperfect sharing of CSI in distributed precoding, it is always assumed that all the TXs designing jointly the precoder have the \emph{same} CSI. This means that practical schemes to reduce the sharing of the CSI with distributed precoding comes down to clustering with the CSI exchange being limited to small \emph{cooperation clusters} inside which the TXs cooperate. The optimal way of forming these clusters has recently become an active research topic\cite{Papadogiannis2008,Papadogiannis2010,Gong2011,Giovanidis2012,Bergel2012}. Still, clustering leads to some fundamental limitations. Firstly, there is inevitably inter-cluster interference on the boundaries of the cluster and secondly, it requires the obtaining at all the TXs inside the cluster of the CSI relative to the entire cluster which means that the amount of CSI feedback required quickly increases with the number of TXs inside the cluster. Several works have focused on determining the optimal size of the clusters when taking into account the cost of estimating the channel elements, e.g., \cite{Huh2012,Lozano2013}. They suggest that TX cooperation cannot efficiently manage interference, even if the backhaul links are strong enough to form large clusters. The main message behind \cite{Lozano2013} is that pilot-based channel estimates incurs a substantial loss when trying to learn the channel from a large number of users within a finite coherence time interval, causing the DoF to saturate. We do not focus in this work on the estimation of the channel but only on the problem of uplink feedback and of CSI sharing between the TXs so that our results do not directly challenge the conclusions of \cite{Lozano2013}  but are in fact complementary. In particular a new perspective arises from the accounting of path loss modeling (and network geometry) in the feedback requirement analysis.

One other important element is that we do not restrict to clustering which has the aforementioned limitations but instead we allow each TX to obtain the CSI relative to any other TX or RX, with the only constraint being on the total amount of information exchanged. Note that the optimization of the feedback allocation in the case where the user's data symbols are \emph{not} shared between the TXs (i.e., coordinated beamforming) yields a completely different problem setting \cite{Zhang2010,Ozbek2010,Saleh2010,Ho2011,Bhagavatula2011a,Tajer2011b,Khoshnevis2012}.

The question that we state is whether it is possible to overcome the fundamental limitations of clustering by optimizing directly the spatial allocation of CSIT. Hence, we study the minimization of the CSI shared across the TXs subject to a given required performance. To tackle this intricate question, we consider the high SNR regime and we study the number of degrees-of-freedom (DoF) achieved. We consider also that the pathloss of the interfering links is parameterized as some function of the SNR. This parameterization allows to model the network geometry and leads to analyze the \emph{generalized DoF} as in\cite{Etkin2008,Tuninetti2007,Wu2007,Jafar2010,Mohapatra2012,Vaze2012b,Chaaban2012}. This modeling of the pathloss as a function of the SNR is essential to model the effect of the network geometry in a DoF analysis where the SNR is assumed to become infinitely large. Indeed, omitting to use such a parameterization and letting the SNR become large makes the pathloss differences (i.e., the network geometry) negligible.

In previous works \cite{dekerret2011_ISIT_journal,dekerret2013_SPAWCspecial}, the DoF and the rate offset have been considered in a distributed CSI setting, where all the wireless links between a TX and a RX have the same pathloss. The focus of \cite{dekerret2011_ISIT_journal} is on the derivation of robust precoders and the approach is completely different due to the restrictive geometry with homogeneous pathloss only.	

In this work, we provide a CSIT dissemination policy, denoted as \emph{distance-based} allowing to achieve the same generalized DoF as that of a cooperative network with perfect CSI at every TX. This CSIT dissemination policy requires only the sharing of the user's data symbols and of the CSI to within a neighborhood which does not increase with the size of the network. Hence, we show that the pathloss attenuation effectively limits the impact of interference to a local neighborhood around each TX and allows for global interference management with only local cooperation.

%The analysis is carried out in a $2$-dimensional network where the \emph{interference level} of a wireless channel, defined as in \cite{Etkin2008}, increases regularly as the distance between the nodes increases. 

\emph{Notations:} We denote by~$\bm{e}_i$ the $i$-th column of the $K\times K$ identity matrix, by~$\bullet^{\He}$ the Hermitian transpose, and by $\delta_{ij}$ the Kronecker symbol which is equal to $1$ if $j=i$ and to zero otherwise. The operator $[\bullet]^{+}$ takes the maximum between the real argument and $0$, and $\lceil\bullet \rceil$ denotes the ceiling operator. $|\mathcal{A}|$ is used to denote the cardinality of the finite set~$\mathcal{A}$. The complex circularly symmetric Gaussian distribution with zero mean and variance~$\sigma^2$ is represented by $\mathcal{N}(0,\sigma^2)$. The $(ij)$-th elements of a matrix~$\mathbf{A}$ is denoted equivalently as $\{\mathbf{A}\}_{ij}$ or as~$A_{ij}$. Let $f$ and $g$ be two functions taking their value in $\mathbb{R}$ with the function $g$ taking only non-zero values. We write $f(x)=o(g(x))$ to denote that $\lim_{x\rightarrow \infty } \frac{|f(x)|}{|g(x)|}=0$. We also use the exponential equality~$f(x)\doteq x^b$ as in \cite{Zheng2003} to denote that $\lim_{x\rightarrow \infty} \frac{\log(f(x))}{\log(x)}=b$.

% 

%%%%%%%%%%%%%%%%%%%%%%%%%%%%%%%%%%%%%%%%%%%%%%%%%%%%%%%%%%%%%%%%%%%%%%%
%%%%%%%%%%%%%%%%%%%%%%%%%%%%%%%%%%%%%%%%%%%%%%%%%%%%%%%%%%%%%%%%%%%%%%%
%%%%%%%%%%%%%%%%%%%%%%%%%%%%%%%%%%%%%%%%%%%%%%%%%%%%%%%%%%%%%%%%%%%%%%%
%%%%%%%%%%%%%%%%%%%%%%%%%%%%%%%%%%%%%%%%%%%%%%%%%%%%%%%%%%%%%%%%%%%%%%%
\section{System Model}\label{se:SM}
%%%%%%%%%%%%%%%%%%%%%%%%%%%%%%%%%%%%%%%%%%%%%%%%%%%%%%%%%%%%%%%%%%%%%%%
%%%%%%%%%%%%%%%%%%%%%%%%%%%%%%%%%%%%%%%%%%%%%%%%%%%%%%%%%%%%%%%%%%%%%%%
%%%%%%%%%%%%%%%%%%%%%%%%%%%%%%%%%%%%%%%%%%%%%%%%%%%%%%%%%%%%%%%%%%%%%%%
%%%%%%%%%%%%%%%%%%%%%%%%%%%%%%%%%%%%%%%%%%%%%%%%%%%%%%%%%%%%%%%%%%%%%%%
% Besides the knowledge of data symbols, each TX is supposed to acquire, through an unspecified feedback or sharing mechanism, its \emph{own} estimate on the channel vectors to all the RXs. 
%%%%%%%%%%%%%%%%%%%%%%%%%%%%%%%%%%%%%%%%%%%%%%%%%%%%%%%%%%%%%%%%%%%%%%%
%%%%%%%%%%%%%%%%%%%%%%%%%%%%%%%%%%%%%%%%%%%%%%%%%%%%%%%%%%%%%%%%%%%%%%%
\subsection{Network MIMO}
%%%%%%%%%%%%%%%%%%%%%%%%%%%%%%%%%%%%%%%%%%%%%%%%%%%%%%%%%%%%%%%%%%%%%%%
%%%%%%%%%%%%%%%%%%%%%%%%%%%%%%%%%%%%%%%%%%%%%%%%%%%%%%%%%%%%%%%%%%%%%%%
We consider a network MIMO setting in which $K$~non-colocated transmitters (TXs) transmit \emph{jointly} via linear precoding to $K$~receivers (RXs) equipped with a single antenna and applying single user decoding. %The extension to multiple-antennas TXs will be discussed later in this work. If the RXs are equipped with multiple-antennas but apply a matched-filter, the analysis directly holds for the virtual MISO channel obtained. Taking into account the zero forcing (ZF) capabilities of the RX is however out of the scope of this work.
Each TX initially has the knowledge of the~$K$ data symbols to transmit to the $K$~RXs (owing to TX cooperation friendly routing protocol for user-plane data). Note that this assumption will be challenged in Section~\ref{se:ExpDec}. The transmission is then described as
\begin{equation}
\begin{bmatrix}
y_1\\\vdots\\y_K
\end{bmatrix} 
=
\begin{bmatrix}
\bm{h}^{\He}_1\\\vdots\\
\bm{h}^{\He}_K
\end{bmatrix}\cdot
\begin{bmatrix}
x_1\\\vdots\\x_K
\end{bmatrix} 
+
\begin{bmatrix}
\eta_1\\\vdots\\
\eta_K
\end{bmatrix} 
\label{eq:SM_1}
\end{equation}% We further define the normalized channel vector $\tilde{\bm{h}}_i\triangleq \bm{h}_i/\norm{\bm{h}_i}$. 
where $\bm{y}\triangleq [y_1,\ldots, y_K]^{\He}\in \mathbb{C}^{K\times 1}$ contains the received signals at the $K$~RXs, $\bm{h}^{\He}_i \in \mathbb{C}^{1\times K}$ is the channel to the $i$-th RX, $\bm{\eta}\triangleq [\eta_1,\ldots,\eta_K]^{\He}\in \mathbb{C}^{K\times 1}$ is the i.i.d.~$\mathcal{N}(0,1)$ normalized noise at the RXs, and $\xv\triangleq [x_1,\ldots,x_K]^{\He}\in \mathbb{C}^{K\times 1}$ represents the transmit signals at the $K$~TXs. 

We also define the multi-user channel~$\mathbf{H}\triangleq [\bm{h}_1,\ldots, \bm{h}_K]^{\He}$. $H_{k,i}$ designates the fading coefficient between TX~$i$ and RX~$k$. We consider a Rayleigh fast fading channel such that~$H_{k,i}=\sigma_{k,i}\tilde{H}_{k,i}$ where $\tilde{H}_{k,i}\sim\mathcal{N}(0,1)$ is a Gaussian random variable and the value of~$\sigma_{k,i}$ will reflect the geometry (topology) of the network. We consider in the following for the sake of clarity that~$\forall i,\sigma_{i,i}^2=1$.

The transmit signal~$\xv$ is obtained from the user's data symbols~$\bm{s}\triangleq[s_1,\ldots,s_K]^{\He}\in\mathbb{C}^{K\times 1}$ (i.i.d. $\mathcal{N}(0,1)$) as
\begin{equation}
\xv=\begin{bmatrix}\bm{t}_1&\ldots& \bm{t}_K\end{bmatrix}\cdot
\begin{bmatrix}s_1\\\vdots\\s_K\end{bmatrix} .
\label{eq:SM_2}
\end{equation}
Hence, the vector~$\bm{t}_i \in \mathbb{C}^{K\times 1}$ represents the beamforming vector used to transmit $s_i$ to RX~$i$ and we define as $\mathbf{T}\triangleq [\bm{t}_1,\ldots, \bm{t}_K]\in \mathbb{C}^{K\times K}$ the multi-user joint precoder. We consider a sum power constraint and an equal power allocation to the users, both for clarity and because it does not impact the DoF. Note that because of the normalization of the noise and of the direct channels, $P$ denotes also the average per-stream SNR. The ergodic rate of user~$i$ is written as
\begin{equation}
R_i\triangleq\E\left[\log_2\left(1+\frac{|\bm{h}_i^{\He}\bm{t}_i|^2}{1+\sum_{\ell\neq i}|\bm{h}_{i}^{\He}\bm{t}_\ell|^2}\right)\right].
\label{eq:SM_3}
\end{equation}
The \emph{DoF} at RX~$i$ is defined as commonly used in the literature as \cite{Tse2005}
\begin{equation}
\DoF_i\triangleq \lim_{P\rightarrow\infty}\frac{R_i}{\log_2(P)}.
\label{eq:SM_4}
\end{equation}
Note that the channel elements are all non-zero with probability one. Consider only these channel realizations does not reduce the DoF achieved but has for consequence that all the expectations can be shown to exist. For the sake of clarity, we consider then only these channel realizations in the following. A more detailed explanation can be found in Appendix~\ref{app:preliminary}.

%%%%%%%%%%%%%%%%%%%%%%%%%%%%%%%%%%%%%%%%%%%%%%%%%%%%%%%%%%%%%%%%%%%%%%%
%%%%%%%%%%%%%%%%%%%%%%%%%%%%%%%%%%%%%%%%%%%%%%%%%%%%%%%%%%%%%%%%%%%%%%%
\subsection{Generalized DoF Analysis}
%%%%%%%%%%%%%%%%%%%%%%%%%%%%%%%%%%%%%%%%%%%%%%%%%%%%%%%%%%%%%%%%%%%%%%%
%%%%%%%%%%%%%%%%%%%%%%%%%%%%%%%%%%%%%%%%%%%%%%%%%%%%%%%%%%%%%%%%%%%%%%%
However, we aim in this work at studying the effect of the network topology in the high SNR performance. In the conventional DoF analysis, any finite pathloss difference is neglected which can lead in some cases to a significant gap between the predicted and the true performance at finite SNR. As a consequence, we use the notion of \emph{generalized DoF}\cite{Etkin2008,Tuninetti2007,Wu2007,Jafar2010,Mohapatra2012,Vaze2012b,Chaaban2012} as a way to obtain a more accurate modeling of the performance. In the generalized DoF approach, the attenuation of the interference is represented as an exponential function of the transmit power so as to preserve the impact of the network geometry in the high SNR analysis. Hence, the generalized DoF at RX~$i$ is defined as
\begin{equation}
\begin{aligned}
\DoF_i(\{\mathbf{B}^{(j)}\}_{j=1}^K,\bm{\Gamma})&\triangleq \lim_{P\rightarrow\infty}\frac{R_i}{\log_2(P)}\\
&\quad\text{subject to $\sigma_{k,i}^2=P^{-\{\bm{\Gamma}\}_{k,i}},\forall k,i$}
\end{aligned}
\label{eq:SM_5}
\end{equation}
where the CSIT allocation~$\{\mathbf{B}^{(j)}\}_{j=1}^K$ and the precoding used will be described in the following and the matrix~$\bm{\Gamma}\in [0,\infty]^{K\times K}$ is called the \emph{interference level matrix} and is given as a function of the parameters of the practical network being studied. Its $(k,i)$-th element is then denoted by~$\Gamma_{k,i}$. More specifically, our goal is to model the transmission in a network with finite pathloss and finite transmit SNR in the most accurate possible way, and the interference level matrix makes the link between the practical network and the model obtained. Denoting by $P_0$ the finite power used in practice and by $\sigma_{k,i,0}$ the variance of the channel coefficient between TX~$i$ and RX~$k$ in the practical setting, the interference level matrix is then defined as
\begin{equation}
\Gamma_{k,i}\triangleq-\frac{\log(\sigma_{k,i,0}^2)}{\log(P_0)},\qquad \forall k,i.
\label{eq:SM_6}
\end{equation}
Note that we assume that all the diagonal coefficients of the interference-level matrix are equal to zero, $\Gamma_{i,i}=0,\forall i$.

As already discussed in the introduction, the goal of the generalized DoF is not to model an unrealistic channel where the pathloss increases with the SNR, in the same way that a DoF analysis does not really apply for transmission with infinite amount of power. It consists simply, starting from a practical setting, in letting both the SNR and the pathloss increase at the same time, instead of letting simply the SNR increase, as in a conventional DoF analysis. This ensures that the differences of power between the wireless links do not become negligible when considering the high SNR regime and hence allows us to take into account the geometry of the network in our analysis. 

\begin{remark}
When $P_0$ -- the transmit power used in the practical setting-- tends to infinity, the prelog factor converges to the (conventional) DoF. The prelog factor converges to the generalized DoF in a different limiting regime where both the pathloss and the SNR increase at the same time. This could for example be the case if the transmit power is made dependent of the distance between the TXs and the RXs. Both approaches however can be used to approximate the performance in practical settings at finite (high) SNR. When there are significant pathloss differences, the generalized DoF will be more accurate.
\qed
\end{remark}

\begin{example} 
Let us consider as toy example a Gaussian IC with two TX/RX pairs and every node having a single-antenna. They interfere to each other via a channel of variance $\sigma^2\in (0,1)$ while the direct links have unit variance. The DoF is well known to be equal to $0.5$ independently of the value of $\sigma^2$\cite{Tse2005}. However, if the TXs interfere with very low power, e.g., $\sigma^2=10^{-12}$, then the interference will be negligible for any realistic range of power used for the transmission. Considering a transmission at SNR~$P_0=30$~dB, the interfering coefficient as defined in \cite{Etkin2008} would be $\alpha=\max(\log(P_0\sigma^2)/\log(P_0),0)=0$ and the generalized DoF would then be equal to $1-\alpha=1$. Hence, the generalized DoF analysis models more accurately the transmission in that setting. \qed
\end{example}

%%%%%%%%%%%%%%%%%%%%%%%%%%%%%%%%%%%%%%%%%%%%%%%%%%%%%%%%%%%%%%%%%%%%%%%
%%%%%%%%%%%%%%%%%%%%%%%%%%%%%%%%%%%%%%%%%%%%%%%%%%%%%%%%%%%%%%%%%%%%%%%
\subsection{Distributed CSI at the TXs}\label{se:SM:CSI}
%%%%%%%%%%%%%%%%%%%%%%%%%%%%%%%%%%%%%%%%%%%%%%%%%%%%%%%%%%%%%%%%%%%%%%%
%%%%%%%%%%%%%%%%%%%%%%%%%%%%%%%%%%%%%%%%%%%%%%%%%%%%%%%%%%%%%%%%%%%%%%% 

The joint precoder is implemented distributively at the TXs with each TX relying solely on its own estimate of the channel matrix in order to compute its transmit coefficient, without any exchange of information with the other TXs\cite{Zakhour2010a,dekerret2011_ISIT_journal}. To model the imperfect CSI at the TX (CSIT), the channel estimate at each TX is assumed to be obtained from a limited rate digital feedback scheme. Consequently, we introduce the following definitions.
\begin{definition}[Distributed Finite-Rate CSIT]
We represent a \emph{CSIT allocation} by the collection of matrices~$\{\mathbf{B}^{(j)}\}_{j=1}^K$ where $\mathbf{B}^{(j)}\in \mathbb{R}_+^{K\times K}$ denotes the CSIT allocation at TX~$j$. Hence, TX~$j$ receives the multi-user channel estimate~$\bH^{(j)}$ defined from
\begin{equation}
H_{k,i}^{(j)}=\sigma_{k,i}\tilde{H}_{k,i}+\sigma_{k,i}\sqrt{2^{-B_{k,i}^{(j)}}}\Delta\tilde{H}^{(j)}_{k,i},\qquad \forall i,k
\label{eq:SM_7}
\end{equation} 
where~$\Delta\tilde{H}^{(j)}_{k,i}\sim\mathcal{N}(0,1)$ and the~$\Delta\tilde{H}^{(j)}_{k,i}$ are mutually independent and independent of the channel. 
\end{definition} 
%The CSIT errors~$\Delta\tilde{H}^{(j)}_{ki}$ are further collected in the matrix~$\bm{\Delta}\tilde{\bH}^{(j)}$ such that
%\begin{equation}
%\tilde{\bH}=\tilde{\bH}^{(j)}+2^{-B_{ki}^{(j)}}\bm{\Delta}\tilde{\bH}^{(j)}.
%\label{eq:SM_9}
%\end{equation} 
\begin{remark}
The reasons for modeling the imperfect CSIT via \eqref{eq:SM_7} are as follows. First, it is well known from rate-distortion theory that the minimal distortion when quantizing a standard Gaussian source using $B$~bits is equal to~$2^{-B}$ \cite[Theorem~$13.3.3$]{Cover2006} while this distortion value is also achieved up to a multiplicative constant using for example the Lloyd algorithm \cite{Girod2013} or even scalar quantization. Thus, the decay in $2^{-B}$ as the number of quantization bits increases, represents a reasonable model. 

Furthermore, only the asymptotic behavior exponentially in the SNR is of interest in this work such that the distribution of the CSIT error does not matter here. We have chosen a Gaussian model for simplicity but other distributions fulfilling some mild regularity constraints could be chosen.\qed
\end{remark}

It is a well known result that the number of CSI feedback bits should scale with the SNR in order to achieve a positive DoF in MISO BCs~\cite{Jindal2006,Caire2010,dekerret2011_ISIT_journal}. Hence, the prelog factor represents an appropriate measure at high SNR of the amount of CSIT required. Thus, we define the \emph{size} of a CSIT allocation as follows.
\begin{definition}[Size of a CSIT allocation]
The size $\Size(\bullet)$ of a CSIT allocation~$\mathbf{B}^{(j)}$ at TX~$j$ is defined as
\begin{equation}
\Size(\mathbf{B}^{(j)})\triangleq \lim_{P\rightarrow \infty} \frac{\sum_{i,k}B_{k,i}^{(j)}}{\log_2(P)}
\label{eq:SM_8}
\end{equation}
such that the total size of a CSIT allocation~$\{\mathbf{B}^{(j)}\}_{j=1}^K$ is
\begin{align}
\Size(\{\mathbf{B}^{(j)}\}_{j=1}^K)&\triangleq \sum_{j=1}^K\Size(\mathbf{B}^{(j)})\\
&=\lim_{P\rightarrow \infty} \frac{\sum_{i,j,k}B_{k,i}^{(j)}}{\log_2(P)}.
\label{eq:SM_9}
\end{align}
\label{def_size}
\end{definition}
\begin{remark} We consider here a digital quantization of the channel vectors but the results can be easily translated to a setting where analog feedback is used. Indeed, digital quantization is simply used as a way to quantify the variance of the CSIT errors \cite{Samardzija2006,Caire2010}. Furthermore, only CSI requirements at the TXs are investigated, and different scenarios can be envisaged for the sharing of the channel estimates (e.g., direct broadcasting from the RXs to all the TXs, sharing through a backhaul, \dots)\cite{dekerret2013_WCM}. 
 \qed
\end{remark}
%%%%%%%%%%%%%%%%%%%%%%%%%%%%%%%%%%%%%%%%%%%%%%%%%%%%%%%%%%%%%%%%%%%%%%%
%%%%%%%%%%%%%%%%%%%%%%%%%%%%%%%%%%%%%%%%%%%%%%%%%%%%%%%%%%%%%%%%%%%%%%%
\subsection{Distributed precoding}\label{se:SM:Precoding}
%%%%%%%%%%%%%%%%%%%%%%%%%%%%%%%%%%%%%%%%%%%%%%%%%%%%%%%%%%%%%%%%%%%%%%%
%%%%%%%%%%%%%%%%%%%%%%%%%%%%%%%%%%%%%%%%%%%%%%%%%%%%%%%%%%%%%%%%%%%%%%%

Based on its individual CSIT, each TX designs its transmit coefficients. We focus here on the CSI dissemination problem under a conventional precoding framework. Hence, we assume that the sub-optimal zero forcing (ZF) precoder is used. Based on its own channel estimate~$\mathbf{H}^{(j)}$, TX~$j$ computes then the ZF beamforming vector~$\bm{t}_i^{(j)}$ to transmit symbol $s_i$ such that 
\begin{equation}
\bm{t}_i^{(j)}\triangleq \sqrt{P}\frac{\left(\mathbf{H}^{(j)}\right)^{-1}\bm{e}_i}{\norm{\left(\mathbf{H}^{(j)}\right)^{-1}\bm{e}_i}},\qquad \forall i\in\{1,\ldots,K\}.
\label{eq:SM_10}
\end{equation} 
\begin{remark}
ZF represents a priori a sub-optimal precoding scheme. It is however well known to achieve the maximal DoF in the MIMO BC with perfect CSIT \cite{Jindal2006,Caire2010}. Furthermore, considering limited feedback in the compound MIMO BC, it is revealed in \cite{Caire2007} that no other precoding scheme can achieve the maximal DoF with a lower feedback scaling. This confirms the efficiency of ZF in terms of DoF, even when confronted with imperfect CSI. ZF represents also the most widely used scheme to manage interference at high SNR.

An exciting yet challenging question is whether there exist strictly better schemes (from a DoF point of view) dealing specifically with the distributed CSI case. This question is however beyond the scope of our work here.
\qed
\end{remark}
%The design of a precoding scheme that is \emph{optimally robust} in the context of distributed CSI is a research topic in its own right, and a challenging one \cite{dekerret2011_ISIT_journal}.
Although a given TX~$j$ may compute the whole precoding matrix $\mathbf{T}^{(j)}$, only the $j$-th row is of practical interest. Indeed, TX~$j$ transmits solely~$x_j=\bm{e}_{j}^{\He}\mathbf{T}^{(j)}\bm{s}$. The effective multi-user precoder~$\mathbf{T}$ verifies then
\begin{equation}
\xv=\mathbf{T}\bm{s}=
\begin{bmatrix}
\bm{e}_1^{\He}\mathbf{T}^{(1)}\\
\bm{e}_2^{\He}\mathbf{T}^{(2)}\\
\vdots\\
\bm{e}_K^{\He}\mathbf{T}^{(K)}
\end{bmatrix}\bm{s}.
\label{eq:SM_11}
\end{equation}
%We denote by the superscript~$\bullet^{\PCSI}$ the coefficients obtained when all the TXs have perfect CSI.

\begin{remark}Each TX independently proceeds with the normalization of the beamformer and based on a-priori different channel estimates. Hence, the power constraint is only approximately fulfilled. Yet, the power constraint is asymptotically fulfilled for all the DoF achieving CSIT allocations that we will consider in the following.
\qed
\end{remark}
Finally, we denote by~$\mathbf{T}^{\star}=[\bm{t}_1^{\star},\ldots,\bm{t}_K^{\star}]$ the precoder obtained with perfect CSI at all TXs. It then verifies
\begin{equation}
\bm{t}_i^{\star}\triangleq \sqrt{P}\frac{\left(\mathbf{H}\right)^{-1}\bm{e}_i}{\norm{\left(\mathbf{H}\right)^{-1}\bm{e}_i}},\qquad \forall i\in\{1,\ldots,K\}.
\label{eq:SM_10}
\end{equation} 
%%%%%%%%%%%%%%%%%%%%%%%%%%%%%%%%%%%%%%%%%%%%%%%%%%%%%%%%%%%%%%%%%%%%%%%
%%%%%%%%%%%%%%%%%%%%%%%%%%%%%%%%%%%%%%%%%%%%%%%%%%%%%%%%%%%%%%%%%%%%%%%
\subsection{Optimization of the CSIT allocation}\label{se:SM:Optimization}
%%%%%%%%%%%%%%%%%%%%%%%%%%%%%%%%%%%%%%%%%%%%%%%%%%%%%%%%%%%%%%%%%%%%%%%
%%%%%%%%%%%%%%%%%%%%%%%%%%%%%%%%%%%%%%%%%%%%%%%%%%%%%%%%%%%%%%%%%%%%%%%

Optimizing directly the allocation of the number of bits at finite SNR represents a challenging problem which gives little hope for analytical results. Instead, we will try to identify one CSIT allocation solution achieving the same DoF as under the fully shared CSIT setting.
\begin{definition}
We define the set of DoF-achieving CSIT allocations~$\mathbb{B}_{\DoF}(\bm{\Gamma})$ as
\begin{equation}
\mathbb{B}_{\DoF}(\bm{\Gamma})\triangleq\{\{\mathbf{B}^{(j)}\}_{j=1}^K|\forall i, \DoF_i(\{\mathbf{B}^{(j)}\}_{j=1}^K,\bm{\Gamma})=1\}.
\label{eq:SM_12}
\end{equation}
\end{definition}
Hence, an interesting problem consists in finding the minimal CSIT allocation (where minimality refers to the size in Definition~\ref{def_size}) which achieves the maximal generalized DoF at every user:
\begin{equation}
\text{minimize   }\Size\LB\{\mathbf{B}^{(j)}\}_{j=1}^K\RB \text{, subject to $\{\mathbf{B}^{(j)}\}_{j=1}^K\in \mathbb{B}_{\DoF}(\bm{\Gamma})$.}	
\label{eq:SM_13}
\end{equation}
In this paper, we focus on an ``achievability" result, by exhibiting a CSIT allocation that achieves the maximal DoF while having a much lower size than the conventional (uniform) CSIT allocation. Furthermore, the proposed ``achievable scheme" will prove to have particularly interesting properties, which distinguish it from other solutions in the literature (e.g., clustering). The problem of finding a minimal-size allocation policy while guaranteeing full DoF (i.e. DoF equal to the perfect CSIT case) is an interesting problem, but an extreme challenging one, which, to our best knowledge, remains open.

%%%%%%%%%%%%%%%%%%%%%%%%%%%%%%%%%%%%%%%%%%%%%%%%%%%%%%%%%%%%%%%%%%%%%%%
%%%%%%%%%%%%%%%%%%%%%%%%%%%%%%%%%%%%%%%%%%%%%%%%%%%%%%%%%%%%%%%%%%%%%%%
%%%%%%%%%%%%%%%%%%%%%%%%%%%%%%%%%%%%%%%%%%%%%%%%%%%%%%%%%%%%%%%%%%%%%%%
%%%%%%%%%%%%%%%%%%%%%%%%%%%%%%%%%%%%%%%%%%%%%%%%%%%%%%%%%%%%%%%%%%%%%%%
\section{Preliminary Results}\label{se:SM:Achievable}
%%%%%%%%%%%%%%%%%%%%%%%%%%%%%%%%%%%%%%%%%%%%%%%%%%%%%%%%%%%%%%%%%%%%%%%
%%%%%%%%%%%%%%%%%%%%%%%%%%%%%%%%%%%%%%%%%%%%%%%%%%%%%%%%%%%%%%%%%%%%%%%
%%%%%%%%%%%%%%%%%%%%%%%%%%%%%%%%%%%%%%%%%%%%%%%%%%%%%%%%%%%%%%%%%%%%%%%
%%%%%%%%%%%%%%%%%%%%%%%%%%%%%%%%%%%%%%%%%%%%%%%%%%%%%%%%%%%%%%%%%%%%%%%

As a preliminary step, we derive a simple sufficient criterion on the precoder for achieving the maximal DoF.
\begin{proposition}
The maximal DoF is achieved by using the precoder~$\mathbf{T}$ if the CSIT allocation~$\{\mathbf{B}^{(j)}\}_{j=1}^K$ is such that 
\begin{equation}
\E \left[\left\| \mathbf{T}-\mathbf{T}^{\star}\right\|_{\Fro}^2\right]\doteq P^{0}
\label{eq:Prelem_1}
\end{equation}
where the equivalence sign $f(P)\doteq P^b$ denotes the exponential equality $\lim_{P\rightarrow \infty} \frac{\log(f(P))}{\log(P)}=b$\cite{Zheng2003} and~$\mathbf{T}^{\star}$ has been defined previously as the precoder based on perfect CSIT.
\label{prop_sufficient}
\end{proposition}
\begin{proof}
A detailed proof is provided in Appendix~\ref{app:proof_sufficient}.
\end{proof} 
The condition obtained above is very intuitive and will be used in the remaining of this work. However, Proposition~\ref{prop_sufficient} does not solve the main question, which is to determine what kind of CSIT allocation allows to achieve~\eqref{eq:Prelem_1}. This question is central to this work and will be tackled in Section~\ref{se:ExpDec}.

%%%%%%%%%%%%%%%%%%%%%%%%%%%%%%%%%%%%%%%%%%%%%%%%%%%%%%%%%%%%%%%%%%%%%%%
\subsection{The conventional CSIT allocation is DoF achieving}\label{se:SM:Conv}
%%%%%%%%%%%%%%%%%%%%%%%%%%%%%%%%%%%%%%%%%%%%%%%%%%%%%%%%%%%%%%%%%%%%%%%
The term ``conventional" hereby corresponds to conveying to each TX the CSI relative to the full multi-user channel, enabling all the TXs to do the same processing and compute a common precoder~$\mathbf{T}^{(j)}=\hat{\mathbf{T}}$. Hence, the condition of Proposition~\ref{prop_sufficient} can be rewritten as
\begin{equation}
\E\left[\left\|\mathbf{T}^{(j)}-\mathbf{T}^{\star}\right\|_{\Fro}^2\right]\doteq P^{0},\qquad \forall j\in \{1,\ldots,K\}.
\label{eq:Prelem_2}
\end{equation}
Based on this, the following result is obtained.
\begin{proposition}
Considering the generalized DoF model where~$\sigma^2_{ki}=P^{-\Gamma_{k,i}},\forall k,i$, the following ``conventional" CSIT allocation~$\{\mathbf{B}^{\Conv, (j)}\}_{j=1}^K$ such that
\begin{align}
\{\mathbf{B}^{\Conv, (j)}\}_{k,i}&= [\lceil \log_2(P\sigma_{k,i}^2) \rceil]^{+},\qquad \forall k,i,j\\
&= \lceil [1-\Gamma_{k,i}]^{+}\log_2(P) \rceil
\label{eq:Prelem_3}
\end{align}  
is DoF achieving, i.e., $\{\mathbf{B}^{\Conv, (j)}\}_{j=1}^K \in \mathbb{B}_{\DoF}$.
\label{prop_conv}
\end{proposition}
\begin{proof}
A detailed proof is provided in Appendix~\ref{app:proof_conv}.
\end{proof}
This CSIT allocation provides to each TX the $K$~channel vectors relative to the $K$~RXs. The term~$\Gamma_{k,i}$ corresponds to the variance of the channel element~$H_{k,i}$ which is equal to~$P^{-\Gamma_{k,i}}$, and follows from well known results of rate-distorsion theory. This means that each TX requires a number of channel estimates growing unbounded with~$K$. This represents a serious issue in large/dense networks which prompts designers, in practice, to restrict cooperation to small cooperations clusters.

%%%%%%%%%%%%%%%%%%%%%%%%%%%%%%%%%%%%%%%%%%%%%%%%%%%%%%%%%%%%%%%%%%%%%%%
\subsection{CSIT allocation with distributed precoding}\label{se:SM:Achievable:Dist}
%%%%%%%%%%%%%%%%%%%%%%%%%%%%%%%%%%%%%%%%%%%%%%%%%%%%%%%%%%%%%%%%%%%%%%%

We now turn our attention to the derivation of a more efficient CSIT allocation strategy. A crucial observation is that each TX does not need to compute accurately the \emph{full} precoder. Indeed, the sufficient criterion~\eqref{eq:Prelem_1} can be written in the distributed CSI setting as
\begin{equation}
\E\left[\left\|\bm{e}_j^{\He}(\mathbf{T}^{(j)}-\mathbf{T}^{\star})\right\|^2\right]\doteq P^{0},\qquad \forall j\in \{1,\ldots,K\}.
\label{eq:Prelem_4}
\end{equation}
Intuition has it that a channel coefficient relative to a TX/RX pair which interferes little with TX/RX~$j$ has little impact on the $j$th precoding row and hence does not need to be known accurately at TX~$j$. What follows is a quantitative assessment of this intuition.

%%%%%%%%%%%%%%%%%%%%%%%%%%%%%%%%%%%%%%%%%%%%%%%%%%%%%%%%%%%%%%%%%%
%%%%%%%%%%%%%%%%%%%%%%%%%%%%%%%%%%%%%%%%%%%%%%%%%%%%%%%%%%%%%%%%%%
%%%%%%%%%%%%%%%%%%%%%%%%%%%%%%%%%%%%%%%%%%%%%%%%%%%%%%%%%%%%%%%%%%
%%%%%%%%%%%%%%%%%%%%%%%%%%%%%%%%%%%%%%%%%%%%%%%%%%%%%%%%%%%%%%%%%%
\section{Distance-Based CSIT Allocation}\label{se:ExpDec}
%%%%%%%%%%%%%%%%%%%%%%%%%%%%%%%%%%%%%%%%%%%%%%%%%%%%%%%%%%%%%%%%%%
%%%%%%%%%%%%%%%%%%%%%%%%%%%%%%%%%%%%%%%%%%%%%%%%%%%%%%%%%%%%%%%%%%
%%%%%%%%%%%%%%%%%%%%%%%%%%%%%%%%%%%%%%%%%%%%%%%%%%%%%%%%%%%%%%%%%%
%%%%%%%%%%%%%%%%%%%%%%%%%%%%%%%%%%%%%%%%%%%%%%%%%%%%%%%%%%%%%%%%%% 
%We consider now a given $2$-dimensional network consisting of $K$~TX/RX pairs for a given interference level matrix~$\bm{\Gamma}$ (defined from the channel variances and the transmit power according to~\eqref{eq:SM_6}), which then determines the value of the $\sigma_{ki}$ from $\sigma_{ki}^2=P^{-\{\bm{\Gamma}\}_{ki}},\forall k,i$. We remind the reader that we have assumed that~$\{\bm{\Gamma}\}_{ii}=1,\forall i$.

%%%%%%%%%%%%%%%%%%%%%%%%%%%%%%%%%%%%%%%%%%%%%%%%%%%%%%%%%%%%%%%%%%
%%%%%%%%%%%%%%%%%%%%%%%%%%%%%%%%%%%%%%%%%%%%%%%%%%%%%%%%%%%%%%%%%% 
\subsection{Distance-based CSIT allocation}\label{se:Exp:CSI}
%%%%%%%%%%%%%%%%%%%%%%%%%%%%%%%%%%%%%%%%%%%%%%%%%%%%%%%%%%%%%%%%%%
%%%%%%%%%%%%%%%%%%%%%%%%%%%%%%%%%%%%%%%%%%%%%%%%%%%%%%%%%%%%%%%%%%  
Before stating our main result, we first define the notion of shortest path which will be needed for the theorem. 
\begin{definition} 
We define a path from TX~$j$ to RX~$k$ as the tuple~$(a_1,\ldots,a_n)$ with $a_i\in\{1,\ldots,K\}$ being the index of a TX/RX pair and $a_1=j$ and $a_n=k$. Given the interference level matrix~$\bm{\Gamma}$, the length~$\mathrm{L}(a_1,\ldots,a_n)$ of a path is then given by 
\begin{equation}
\mathrm{L}(a_1,\ldots,a_n)=\sum_{i=1}^{n-1}\Gamma_{a_{i+1},a_{i}}.
\label{eq:exp_0}
\end{equation}
We can then define the shortest path from TX~$j$ to RX~$k$, which we denote by~$\Gamma_{j \rightarrow k}$, in the sense that
\begin{equation}
\Gamma_{j\rightarrow k}\triangleq \min_{(a_2\ldots,a_{n-1})} \mathrm{L}(j,a_2\ldots,a_{n-1},k).
\label{eq:exp_1}
\end{equation}
%\begin{equation}
%\begin{aligned}
%\Gamma_{k\rightarrow j}&\triangleq \min_{(k,a_2\ldots,a_{n-1},j)} \mathrm{L}(k,a_2\ldots,a_{n-1},j)\\
%&~~~~~~ \text{subject to $a_1=k$ and $a_n=j$.}
%\end{aligned}
%\label{eq:exp_1}
%\end{equation}
%
\end{definition}  
\begin{example}
Let us consider the shortest path~$\Gamma_{1\rightarrow 3}$ in a network with $3$ TX/RX pairs. Keeping in mind that~$\Gamma_{i,i}=0,\forall i$, it is then simply equal to 
\begin{equation}
\Gamma_{1\rightarrow 3}=\min\LB \Gamma_{3,2}+\Gamma_{2,1},\Gamma_{3,1}\RB.
\end{equation}
\qed
\end{example}
We can now state our first main result.
\begin{theorem}
\label{thm}
Let us define the CSIT allocation~$\{\mathbf{B}^{\Dist,(j)}\}_{j=1}^K$ such that 
\begin{align} 
\{\mathbf{B}^{\Dist,(j)}\}_{k,i}&\triangleq \lceil [1-\Gamma_{k,i}-\gamma^{(j)}_{k,i}]^{+}\log_2(P) \rceil ,\qquad \forall k,i,j 
\label{eq:exp_2}
\end{align}  
with 
\begin{equation}
\gamma^{(j)}_{k,i}\triangleq \min\LB\Gamma_{k\rightarrow j},\min_{\ell}\Gamma_{\ell \rightarrow i}+\Gamma_{j,\ell}\RB
\label{eq:exp_3}
\end{equation}
Then~$\mathbf{B}^{\Dist}\in \mathbb{B}_{\DoF}$.
\end{theorem}  
\begin{proof} 
A detailed proof is provided in Appendix~\ref{app:proof_thm}.
\end{proof} 
If the interference level matrix~$\bm{\Gamma}$ tends to the zero matrix~$\bm{0}_{K}$, there is then no attenuation of the interference due to the pathloss and the distance-based CSIT allocation converges as expected to the conventional CSIT allocation given in~\eqref{eq:Prelem_3}. More generally, the distance-based CSIT allocation exploits the fact that if two TX/RX pairs interfere only through weak channels, they need to exchange a small amount of CSI. 
\begin{remark}
If the interference level matrix is a symmetric matrix, $\gamma^{(j)}_{k,i}$ is then equal to
\begin{equation}
\gamma^{(j)}_{k,i}=\min\LB\Gamma_{k\rightarrow j},\Gamma_{i \rightarrow j}\RB
\label{eq:exp_4}
\end{equation}
Another interesting case arises if the interference level matrix satisfies that
\begin{equation}
\Gamma_{k,j}\leq \Gamma_{k,i}+\Gamma_{i,j},\qquad \forall i,j,k
\label{eq:exp_5}
\end{equation}
It then holds that
\begin{equation}
\Gamma_{k\rightarrow j}=\Gamma_{j,k} .
\label{eq:exp_6}
\end{equation}
In particular, the relation \eqref{eq:exp_5} is satisfied when long term attenuation introduces a notion of distance between the TX/RX pairs.
\qed
\end{remark} 
Building upon the proof of Theorem~\ref{thm}, it is also possible to comment equation~\eqref{eq:exp_2} to obtain interesting insights:
\begin{itemize}
\item The term $\Gamma_{k,i}$ follows from the variance of the element to quantize and is also present in the conventional CSIT allocation. Hence, the CSIT reduction comes from the parameter~$\gamma_{k,i}^{(j)}$. 
\item The first term~$\Gamma_{k\rightarrow j}$ corresponds to a sufficient CSIT allocation such that~$\bm{e}_j^{\He}\mathbf{H}^{-1}\bm{e}_i$ is known at TX~$j$ with a sufficient accuracy for every $i$.
\item With the second term~$\min_{\ell}\Gamma_{\ell \rightarrow i}+\Gamma_{j, \ell}$, the CSIT allocation obtained ensures that the norm~$\|\bH^{-1}\bm{e}_i\|$ is known for every~$i$ with a sufficient accuracy at TX~$j$.
\end{itemize}
Taken together (which explains the $\min$ in \eqref{eq:exp_3}), these two requirements allow to compute~$\bm{e}_j^{\He}\bH^{-1}\bm{e}_i/\norm{\bH^{-1}\bm{e}_i}=\bm{e}_j^{\He}\bm{t}_i$ at TX~$j$ for every~$i$ with an accuracy sufficient for achieving the maximal DoF at each RX.
%%%%%%%%%%%%%%%%%%%%%%%%%%%%%%%%%%%%%%%%%%%%%%%%%%%%%%%%%%%%%%%%%%
%%%%%%%%%%%%%%%%%%%%%%%%%%%%%%%%%%%%%%%%%%%%%%%%%%%%%%%%%%%%%%%%%% 
\subsection{Scaling properties of the Distance-based CSIT allocation} \label{se:Exp:Properties}
%%%%%%%%%%%%%%%%%%%%%%%%%%%%%%%%%%%%%%%%%%%%%%%%%%%%%%%%%%%%%%%%%%
%%%%%%%%%%%%%%%%%%%%%%%%%%%%%%%%%%%%%%%%%%%%%%%%%%%%%%%%%%%%%%%%%% 
For clarity, we consider in the following that the interference level matrix is symmetric such that~$\Gamma_{k,i}=\Gamma_{i,k},\forall k,i$. It also corresponds to most of the practically relevant scenarios. Similar conclusions also hold in the asymmetric scenarios.

An important property of a CSIT allocation is its scaling behaviour as the number of TX/RX pairs increases. It is clear that this property depends on the network geometry. For example, if all the TX/RX pairs are collocated, full CSIT sharing is required while it may not be the case in typical network deployment scenarios, as shown previously. We focus now on the setting where each RX is ``interfered" by only a finite number of TXs with an average power larger than $1/P$, who represent the significant interferers in terms of DoF. That is to say
\begin{equation}
\lim_{K\rightarrow \infty} \left|\left\{i|\bm{\Gamma}_{i\rightarrow j}<1\right\}\right|<\infty,\qquad \forall j.
\label{eq:exp_7}
\end{equation}
We will show in the following subsection that this condition is fulfilled in the networks practically encountered in wireless communications.
\begin{corollary} 
Let us consider a symmetric network where condition \eqref{eq:exp_7} is satisfied. It then holds that
\begin{equation}
\lim_{K\rightarrow \infty} \Size(\mathbf{B}^{\Dist,(j)})<\infty,\qquad \forall j.
\label{eq:exp_8}
\end{equation} 
\end{corollary}  
\begin{proof}
This result follows directly from observing that $\Gamma_{k,i}\geq \Gamma_{i \rightarrow k}$ in the distance-based CSIT allocation in \eqref{eq:exp_2}. It follows then trivially from \eqref{eq:exp_7} that there are only a finite number of nonzero $\{\mathbf{B}^{\Dist,(j)}\}_{k,i}$ at every TX~$j$.
\end{proof}
This result is in stark contrast with the conventional CSIT allocation where the size~$\Size(\mathbf{B}^{\Conv,(j)})$ scales linearly with~$K$. This corollary confirms the intuition that a CSIT-exchange restricted to a finite neighborhood is sufficient to achieve global coordination from a DoF point of view. 

Yet, we have considered so far a scenario with global sharing of the user's data symbols. The result above leads to ask ourselves whether this assumption is necessary or if it is possible to reduce the sharing of the user's data symbols without impacting the DoF achieved.
\begin{corollary} 
\label{corollary_data_Exp}
Let us consider a symmetric network where condition \eqref{eq:exp_7} is satisfied. Let us further denote by $\mathcal{K}_j$ the set containing the user's data symbols which have to be shared to TX~$j$ in order to achieve the maximal DoF at every RX. It then holds that   
\begin{equation}
\lim_{K\rightarrow \infty}|\mathcal{K}_j|<\infty,\qquad \forall j\in \{1,\ldots, K\}.
\label{eq:exp_9}
\end{equation}  
\end{corollary}
\begin{proof}
It can be seen from the expression of the precoder coefficients as an infinite summation in the proof of Theorem~\ref{thm} in Appendix~\ref{app:proof_thm} that
\begin{equation}
\E[|\bm{e}_j^{\He}\bH^{-1}\bm{e}_i|^2]\dotleq P^{- \Gamma_{j,i}},\qquad \forall i,j.
\label{eq:exp_10}
\end{equation}   
Setting to~$0$ all the coefficients in the precoder with an exponent of $P$ smaller than~$-1$ leads to additional interferences which tend to zero as the SNR increases and are hence negligible in terms of DoF. If the coefficient $(j,i)$ is set to~$0$ at TX~$j$, this means that TX~$j$ does not need to receive the user's data symbol~$s_i$. From the assumption \eqref{eq:exp_7}, there are then only a finite number of user's data symbols which need to be known at a given TX.
\end{proof}
%\begin{remark}
%The results can in fact be adapted to the case where 
%\begin{equation}
%\lim_{K\rightarrow \infty}\sum_{j=1}^K[1-\{\bm{\Gamma}\}_{ij}]^+<\infty,\qquad \forall i.
%\label{eq:exp_5}
%\end{equation}
%Since condition \eqref{eq:exp_5} holds in most of the practical settings encountered as a consequence of the pathloss, we have restricted our analysis to this case.
%\qed
%\end{remark}
The operational meaning of the above result is that because of the pathloss attenuation it is possible to achieve DoF-perfect coordination on the TX side with each TX exchanging information (CSI or data symbol) only to a local neighborhood.

%%%%%%%%%%%%%%%%%%%%%%%%%%%%%%%%%%%%%%%%%%%%%%%%%%%%%%%%%%%%%%%%%%
%%%%%%%%%%%%%%%%%%%%%%%%%%%%%%%%%%%%%%%%%%%%%%%%%%%%%%%%%%%%%%%%%% 
\subsection{Scaling Behaviour in Wireless Networks} \label{se:Exp:Polynomial}
%%%%%%%%%%%%%%%%%%%%%%%%%%%%%%%%%%%%%%%%%%%%%%%%%%%%%%%%%%%%%%%%%%
%%%%%%%%%%%%%%%%%%%%%%%%%%%%%%%%%%%%%%%%%%%%%%%%%%%%%%%%%%%%%%%%%%  
We have studied above the properties of the transmission for a given interference level matrix~$\bm{\Gamma}$. We will now show how these results can be used to model the transmission in realistic settings. The first step is to discuss how the interference-level matrix is obtained from the network configuration. Let us consider a network with a polynomial pathloss with exponent~$\varepsilon>0$, which corresponds to the conventional model for wireless networks\cite{Tse2005}. This means that the long term attenuation between TX~$j$ and RX~$i$ is equal to $d_{i,j}^{-\varepsilon}$ where $d_{i,j}$ is the distance separating TX~$j$ and RX~$i$. The operational SNR is set to be $P_0$ such that we obtain that
\begin{equation} 
\Gamma_{i,j}=-\frac{\log(d_{i,j}^{-\varepsilon})}{\log(P_0)},\qquad \forall i,j.
\label{eq:exp_11}
\end{equation}
%As $d_{ij}$ increases, $\{\bm{\Gamma}\}_{ij}$ tends to $-\infty$ such that $[1-\{\bm{\Gamma}\}_{ij}]^+$ will be equal to zero. 
Upon defining
\begin{equation}
d_0\triangleq P_0^{\frac{1}{\varepsilon}},
\label{eq:exp_12}
\end{equation}
we can see that
\begin{equation}
\Gamma_{i,j}>1,\qquad \text{if $d_{ij}>d_0$.}
\label{eq:exp_13}
\end{equation}

We are particularly interested in this section in the scaling behaviour as the number of TX/RX pairs increases. This requires defining more precisely the spatial distribution of the TX/RX pairs. It is differentiated in the literature between so-called \emph{dense} networks and \emph{extended} networks\cite{Xie2004,Ozgur2007}. In the first model, the size of the network remains constant and the density (number of TX/RX pairs/$m^2$) increases, while in the second the density of the network remains constant as the number of TX/RX pairs increases. Our analysis being on large networks, we consider the extended model and we assume that the density of TX/RX pairs remains constant.

It follows from \eqref{eq:exp_13} and from the extended model that the necessary condition \eqref{eq:exp_7} is fulfilled in this network configuration. Hence, the corollaries provided above can be applied. Practically, this means that it is possible to achieve the performance of global cooperation with cooperation of the TXs restricted to a local scale. Altogether, the distance-based CSIT allocation along with the matching limited user’s data sharing provides an attractive alternative to clustering. The difference being that the hard-boundaries of the clusters are replaced by a smooth decrease of the level of cooperation.
%%%%%%%%%%%%%%%%%%%%%%%%%%%%%%%%%%%%%%%%%%%%%%%%%%%%%%%%%%%%%%%%%%%%%%%
%%%%%%%%%%%%%%%%%%%%%%%%%%%%%%%%%%%%%%%%%%%%%%%%%%%%%%%%%%%%%%%%%%%%%%%
%%%%%%%%%%%%%%%%%%%%%%%%%%%%%%%%%%%%%%%%%%%%%%%%%%%%%%%%%%%%%%%%%%%%%%%
%%%%%%%%%%%%%%%%%%%%%%%%%%%%%%%%%%%%%%%%%%%%%%%%%%%%%%%%%%%%%%%%%%%%%%%
\section{Simulations}\label{se:sim} 
%%%%%%%%%%%%%%%%%%%%%%%%%%%%%%%%%%%%%%%%%%%%%%%%%%%%%%%%%%%%%%%%%%%%%%%
%%%%%%%%%%%%%%%%%%%%%%%%%%%%%%%%%%%%%%%%%%%%%%%%%%%%%%%%%%%%%%%%%%%%%%%  
%%%%%%%%%%%%%%%%%%%%%%%%%%%%%%%%%%%%%%%%%%%%%%%%%%%%%%%%%%%%%%%%%%%%%%%
%%%%%%%%%%%%%%%%%%%%%%%%%%%%%%%%%%%%%%%%%%%%%%%%%%%%%%%%%%%%%%%%%%%%%%%
 
We verify now by simulations that the maximal DoF per user is achieved by the distance based CSIT allocation. At the same time, we compare the distance based CSIT allocation to the CSI disseminations commonly used, i.e., uniform CSIT allocation and clustering. 

We consider a wireless model with polynomial attenuation as described in Subsection~\ref{se:Exp:Polynomial}. We choose~$\varepsilon=2$, $P_0=30$dB and the interference-level matrix~$\bm{\Gamma}$ is obtained from \eqref{eq:exp_11}. Note that we consider only $d_{i,j}>1$ to ensure that the interfering links are weaker than the direct links. We use Monte-Carlo averaging over $1000$~channel realizations. 

In a first step, we study a network with a regular geometry where $K=36$~TX/RX pairs are placed at the integer values inside a square of dimensions~$6\times 6$. We show in Fig.~\ref{R_square_K36} the average rate achieved with different CSIT allocation policies. Specifically, the distance-based CSIT allocation in \eqref{eq:exp_5} is compared to two alternative CSIT allocations, being the uniform CSIT allocation~$\{\bB^{\mathrm{unif},(j)}\}_{j=1}^K$ where the bits are allocated \emph{uniformly} to the TXs, and the clustering one~$\{\bB^{\mathrm{cluster},(j)}\}_{j=1}^K$ in which (non-overlapping) \emph{regular clustering of size~$4$} is used. Both CSIT allocations are chosen to have the same size as the distance-based one:
\begin{equation}
\Size\!\LB\!\{\bB^{\mathrm{unif},(j)}\}_{j=1}^K \!\RB\!=\!\Size\!\LB\{\bB^{\mathrm{cluster},(j)}\}_{j=1}^K\!\RB\!=\!\Size\!\LB\!\{\bB^{(\Dist,(j)}\}_{j=1}^K\!\RB\!.
\end{equation}

\begin{figure}[htp!] 
\centering
\includegraphics[width=1\columnwidth]{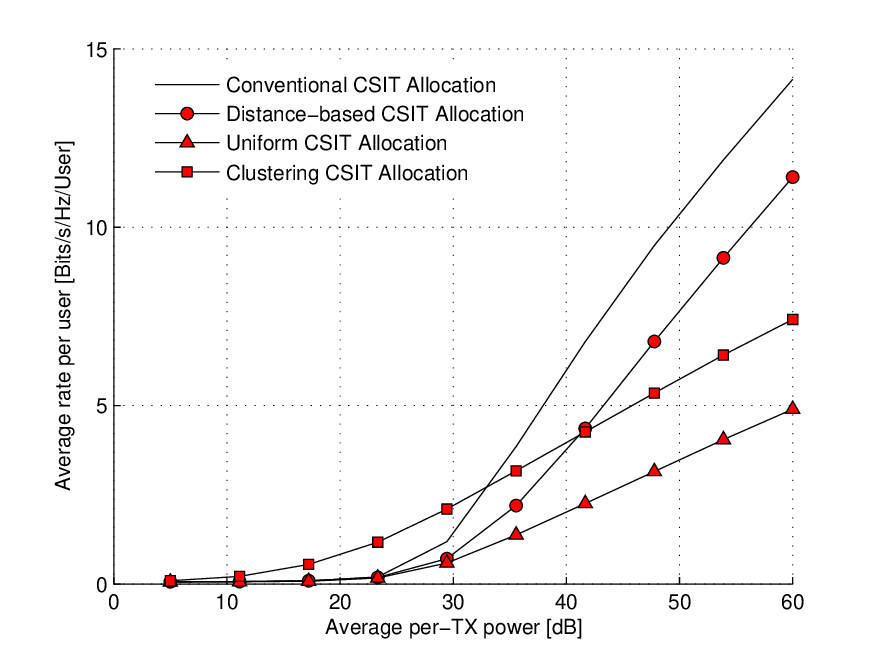}
\caption{Average rate per user as a function of the SNR~$P$ for $K=36$ with the polynomial model described in Subsection~\ref{se:Exp:Polynomial}. The TX/RX pairs are positioned at the integers values inside a square of dimensions~$6\times 6$. The $3$ limited feedback CSIT allocations used have the same size which is equal to~$9\%$ of the size of the conventional CSIT allocation in~\eqref{eq:Prelem_3}.}
\label{R_square_K36}
\end{figure}

With these parameters, the size of the distance based CSIT allocation is only equal to $9\%$ of the size of the conventional CSIT allocation. Nevertheless, it can be observed to achieve the maximal generalized DoF while the clustering solution has a smaller slope. The distance-based CSIT allocation suffers from a strong negative rate offset. However, this offset is a consequence of our analysis being limited to the high SNR regime and can be also observed in the fact that clustering outperforms ZF based on the conventional CSIT-allocation, which represents in fact the true reference for our scheme. Indeed, using ZF with many users is very inefficient at intermediate SNR, particularly in a network with strong pathloss. Furthermore, the number of TX/RX pairs~$K$ which is here relatively large, has not been taken into account. Hence, this strong negative rate offset can be easily reduced by optimizing the precoding scheme and the CSIT allocation at finite SNR. The key element being that the distance-based CSIT allocation does not present the usual limitations of clustering, i.e., edge-interference and bad scaling properties as the size of the cluster increases.

Finally, we show in Fig.~\ref{R_random_K15} the average rate per user in a network made of $K=15$ TX/RX pairs being located uniformly at random over the same square of dimensions~$6\times 6$. To verify the impact of allocating more --or less-- CSIT, we compare the average rate achieved with the distance-based CSIT allocation to the average rate obtained if we use the following variation of the CSIT allocation:
\begin{align} 
\{\mathbf{B}^{\Dist,(j)}\}_{k,i}(\alpha)\!=\! \lceil [1-\Gamma_{j,k}-\alpha \gamma^{(j)}_{k,i}]^{+}\log_2(P) \rceil,\forall k,i,j.
\label{eq:sim_1}
\end{align}  
This allows to observe the impact of reducing ($\alpha>1$) or increasing ($\alpha<1$) the CSIT compared to the distance-based CSIT allocation.

We can observe that reducing the CSIT allocation leads to reducing the slope, i.e., the DoF, while using more feedback bits leads to a vanishing rate offset. This is in agreement with our theoretical result that the distance-based CSIT allocation leads to a finite (bounded) rate offset.
 
 \begin{figure}[htp!] 
\centering
\includegraphics[width=1\columnwidth]{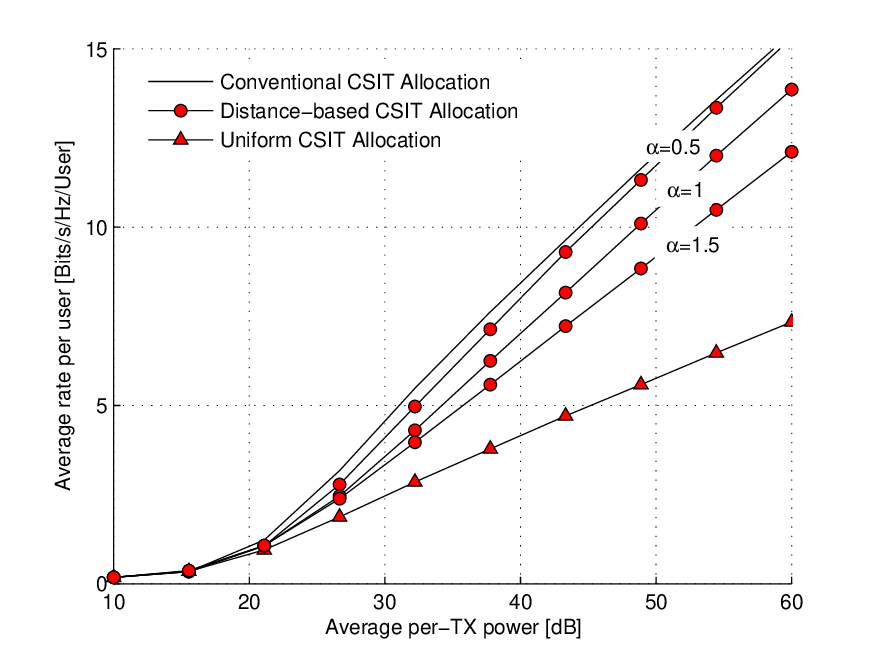}
\caption{Average rate per user as a function of the SNR~$P$ for $K=15$ for the channel model described in Subsection~\ref{se:Exp:Polynomial}. The TX/RX pairs are placed uniformly at random over a square of dimensions~$6\times 6$. The CSIT allocation with $\alpha=0.5$, $\alpha=1$, and $\alpha=1.5$, use respectively~$39\%$, $17\%$, and $13\%$ of the number of bits relative to the conventional CSIT allocation in~\eqref{eq:Prelem_3}.}
\label{R_random_K15}
\end{figure}

\FloatBarrier

%%%%%%%%%%%%%%%%%%%%%%%%%%%%%%%%% 
%%%%%%%%%%%%%%%%%%%%%%%%%%%%%%%%% 
\section{Conclusion}\label{se:conclusion}
%%%%%%%%%%%%%%%%%%%%%%%%%%%%%%%%% 
%%%%%%%%%%%%%%%%%%%%%%%%%%%%%%%%% 
We have discussed the problem of optimizing the CSIT dissemination in a network MIMO scenario. In particular, following a generalized DoF analysis, we have exhibited a CSIT allocation which allows to achieve the optimal generalized DoF while restricting the cooperation to a local scale. This behavior is critical for the cooperation of a large number of TXs to be practical. Hence, the proposed CSIT allocation appears as an alternative to clustering where the hard boundaries of the cluster are replaced by a smooth decrease of the cooperation strength. Our focus has been on the high SNR performance, and the distance-based CSIT allocation should be further optimized to lead to gain in realistic transmissions. Yet, it appears to have a strong potential as an alternative to clustering. In addition, we have considered only the CSIT requirements in order to achieve global interference management. The design of feedback schemes and backhaul links allowing to achieve these requirements represents another very interesting research area.  

This work shows that the CSIT requirements do not have to scale unbounded with the size of the network, which differs from the conclusions of several works from the literature. This is a consequence from letting the pathloss increase with the SNR, which makes the pathloss non-negligible at high SNR. We believe that this is the proper modelization of the pathloss in order to keep the impact of the network geometry, which, in contrast, becomes negligible in a DoF analysis with fixed pathloss.

%%%%%%%%%%%%%%%%%%%%%%%%%%%%%%%%%%%%%%%%%%%%%%%%%%%%
%%%%%%%%%%%%%%%%%%%%%%%%%%%%%%%%%%%%%%%%%%%%%%%%%%%%
%%%%%%%%%%%%%%%%%%%%%%%%%%%%%%%%%%%%%%%%%%%%%%%%%%%%
%%%%%%%%%%%%%%%%%%%%%%%%%%%%%%%%%%%%%%%%%%%%%%%%%%%%
\appendices
%%%%%%%%%%%%%%%%%%%%%%%%%%%%%%%%%%%%%%%%%%%%%%%%%%%%
%%%%%%%%%%%%%%%%%%%%%%%%%%%%%%%%%%%%%%%%%%%%%%%%%%%%
%%%%%%%%%%%%%%%%%%%%%%%%%%%%%%%%%%%%%%%%%%%%%%%%%%%%
%%%%%%%%%%%%%%%%%%%%%%%%%%%%%%%%%%%%%%%%%%%%%%%%%%%%

%%%%%%%%%%%%%%%%%%%%%%%%%%%%%%%%%%%%%%%%%%%%%%%%%%%%
%%%%%%%%%%%%%%%%%%%%%%%%%%%%%%%%%%%%%%%%%%%%%%%%%%%%
\section{Preliminary remark}\label{app:preliminary}
%%%%%%%%%%%%%%%%%%%%%%%%%%%%%%%%%%%%%%%%%%%%%%%%%%%%
%%%%%%%%%%%%%%%%%%%%%%%%%%%%%%%%%%%%%%%%%%%%%%%%%%%%
We start by a remark which will allow us to consider all the expectations of channel elements as finite. For a given $\varepsilon>0$, let us consider the channels satisfying
\begin{align}
\varepsilon<|H_{i,k}|^2<\frac{1}{\varepsilon},\qquad \forall k,i
\label{eq:General_proof_1} 
\end{align} 
such that the expectations of fractions of channel coefficients will all be finite. The channel elements being Chi-$2$ distributed, their are distributed such that
\begin{align}
\Pr\{|H_{i,k}|^2<x\}=1-\exp\LB-\frac{x}{2}\RB.
\label{eq:General_proof_2} 
\end{align}
Hence, considering only the channels verifying \eqref{eq:General_proof_1} and assuming the DoF loss to be maximal for the other channel realizations (i.e., equal to~$1$ for each user) can be easily shown to lead to a DoF loss in the order of $(O(\varepsilon))$. Letting $\varepsilon$ tend to zero, the $(O(\varepsilon))$ term vanishes. Hence, considering only the channels verifying the above condition does not modify the DoF and we will in the following consider only such channel realizations.

%%%%%%%%%%%%%%%%%%%%%%%%%%%%%%%%%%%%%%%%%%%%%%%%%%%%
%%%%%%%%%%%%%%%%%%%%%%%%%%%%%%%%%%%%%%%%%%%%%%%%%%%%
\section{Proof of Proposition~\ref{prop_sufficient}}\label{app:proof_sufficient}
%%%%%%%%%%%%%%%%%%%%%%%%%%%%%%%%%%%%%%%%%%%%%%%%%%%%
%%%%%%%%%%%%%%%%%%%%%%%%%%%%%%%%%%%%%%%%%%%%%%%%%%%%
\begin{proof}
We start by defining the rate difference $\Delta_{\Rate,i}$ between the rate of user~$i$ based on perfect CSI and the rate achieved with limited feedback. %Remember he normalized beamformer~$\bm{t}_i\triangleq \bm{t}_i/\norm{\bm{t}_i}$ with limited feedback and its counter-art~$\bm{t}_i^{\star}\triangleq \bm{t}_i^{\star}/\norm{\bm{t}_i^{\star}}$ based on CSI. 
As in \cite{Jindal2006,Caire2010}, we can then write
\begin{align}
&\Delta_{\Rate,i}\notag\\
&\!\triangleq\!\E \!\bigg[\log_2(1\!+\!|\bm{h}_i^{\He}\bm{t}^{\star}_i|^2)\bigg] \!-\!\E \bigg[\log_2\bigg(\!1\!+\!\frac{|\bm{h}_i^{\He}\bm{t}_i|^2}{1\!+\!\sum_{j\neq i}|\bm{h}_i^{\He}\bm{t}_j|^2}\!\bigg)\!\bigg]\label{eq:roof_Sufficient_1_1}\\
&\!=\!\E\bigg[\log_2\!\bigg(\!\frac{1\!+\!|\bm{h}_i^{\He}\bm{t}^{\star}_i|^2}{1\!+\!\sum_{j=1}^K |\bm{h}_i^{\He}\bm{t}_j|^2}\!\bigg)\!\bigg] \!+\!\E\!\bigg[\log_2\!\bigg(1\!+\!\sum_{j\neq i}|\bm{h}_i^{\He}\bm{t}_j|^2\!\bigg)\!\bigg]\label{eq:roof_Sufficient_1_2}\\
&=\!\E\!\bigg[\log_2\!\bigg(1\!+\!\sum_{j\neq i}|\bm{h}_i^{\He}\bm{t}_j|^2\!\bigg)\!\bigg]+o(\log(P))
\label{eq:roof_Sufficient_1_3}
\end{align}  
where we have denoted by~$\bm{t}^{\star}_i$ the $i$th ZF beamformer based on perfect CSIT. We further obtain
\begin{align}
\Delta_{\Rate,i}&\!=\!\E\!\bigg[\log_2\!\bigg(\!1\!+\!\sum_{j\neq i}|\bm{h}_i^{\He}(\bm{t}_j^{\star}+(\bm{t}_j-\bm{t}_j^{\star}))|^2\!\bigg)\!\bigg]\!+\!o(\log(P))\label{eq:roof_Sufficient_2_1}\\
%&=\!\E\!\bigg[\log_2\!\bigg(1\!+\!\sum_{j\neq i}(|\bm{h}_i^{\He}\bm{t}_j^{\star}|^2+|\bm{h}_i^{\He}(\bm{t}_j-\bm{t}_j^{\star})|^2)\!\bigg)\!\bigg]\label{eq:roof_Sufficient_2_2}\\
&=\!\E\!\bigg[\log_2\!\bigg(1\!+\!\!\sum_{j\neq i} |\bm{h}_i^{\He}(\bm{t}_j-\bm{t}_j^{\star})|^2)\!\bigg)\!\bigg]\!+\!o(\log(P)).
\label{eq:roof_Sufficient_2_3}
\end{align}  
We can then easily upper-bound \eqref{eq:roof_Sufficient_2_3} to write 
\begin{align}
\Delta_{\Rate,i}&\leq \E \bigg[\log_2\bigg(1+\norm{\bm{h}_i}^2\sum_{j\neq i}\norm{\bm{t}_j-\bm{t}_j^{\star}}^2\bigg)\bigg]+o(\log(P))\label{eq:roof_Sufficient_3_1}\\
&\stackrel{(a)}{\leq}\E\bigg[\log_2\bigg(1+ \norm{\mathbf{T}-\mathbf{T}^{\star}}_{\Fro}^2\bigg)\bigg] \notag \\
&\qquad\qquad+\E\bigg[\log_2\bigg(1+\norm{\bm{h}_i}^2\bigg)\bigg]+o(\log(P)) \label{eq:roof_Sufficient_3_2}\\
&\stackrel{(b)}{\leq} \log_2\bigg(\E\bigg[\norm{\mathbf{T}-\mathbf{T}^{\star}}_{\Fro}^2\bigg]\bigg)\!+\!o(\log(P))\label{eq:roof_Sufficient_3_3}\\
&\leq o(\log(P))
\label{eq:roof_Sufficient_3_4}
\end{align}	
where inequality $(a)$ follows from the property that for $a,b\geq 0$, then $\log(1+ab)\leq \log(1+a)+\log(1+b)$ and inequality $(b)$ from the assumption that~$\E\bigg[\norm{\mathbf{T}-\mathbf{T}^{\star}}_{\Fro}^2\bigg]\doteq P^{0}$. The maximal DoF is achieved if the rate difference $\Delta_{\Rate,i}/\log_2(P)$ tends to zero as the SNR increases, which is exactly what has been demonstrated above.
\end{proof}

%%%%%%%%%%%%%%%%%%%%%%%%%%%%%%%%%%%%%%%%%%%%%%%%%%%%
%%%%%%%%%%%%%%%%%%%%%%%%%%%%%%%%%%%%%%%%%%%%%%%%%%%%
\section{Proof of Proposition~\ref{prop_conv}}\label{app:proof_conv}
%%%%%%%%%%%%%%%%%%%%%%%%%%%%%%%%%%%%%%%%%%%%%%%%%%%%
%%%%%%%%%%%%%%%%%%%%%%%%%%%%%%%%%%%%%%%%%%%%%%%%%%%%

\begin{proof}
We consider without loss of generality the precoding at TX~$j$. Using the CSIT allocation in \eqref{eq:Prelem_3}, it holds that
\begin{equation}
\sigma_{k,i}\sqrt{2^{-B_{k,i}^{(j)}}}=\sqrt{\frac{1}{P}}	
\label{eq:prof_Conv_1}
\end{equation}
such that~$\bH^{(j)}=\bH+\sqrt{\frac{1}{P}}	\bm{\Delta}\bH^{(j)}$. We start by recalling the well known resolvent equality.
\begin{proposition}[Resolvent equality]
Let $\mathbf{A}\in \mathbb{C}^{n\times n}$ and $\mathbf{B}\in \mathbb{C}^{n\times n}$ be two invertible matrices, it then holds that
\begin{equation}
\mathbf{A}^{-1}-\mathbf{B}^{-1}=\mathbf{B}^{-1}(\mathbf{B}-\mathbf{A})\mathbf{A}^{-1}.
\end{equation}
\end{proposition}
Using the resolvent equality two times successively, we can write
\begin{align}
&\left(\mathbf{H}^{(j)}\right)^{-1}\!\!\!\!\!\!-\mathbf{H}^{-1}\notag\\
&=\left(\bH+\sqrt{P^{-1}}	\bm{\Delta}\bH^{(j)}\right)^{-1}-\mathbf{H}^{-1}\\
&=\mathbf{H}^{-1} (-\sqrt{P^{-1}}	\bm{\Delta}\bH^{(j)})(\mathbf{H}^{(j)})^{-1}\\
&=\mathbf{H}^{-1} (-\sqrt{P^{-1}} \bm{\Delta}\bH^{(j)})\bH^{-1}\notag\\
&\qquad+\mathbf{H}^{-1} (-\sqrt{P^{-1}} \bm{\Delta}\bH^{(j)})\LB (\mathbf{H}^{(j)})^{-1}-\bH^{-1}\RB\\
&=-\sqrt{P^{-1}}\mathbf{H}^{-1}\bm{\Delta}\bH^{(j)} \mathbf{H}^{-1}\notag\\
&\qquad+P^{-1}\mathbf{H}^{-1}\bm{\Delta}\bH^{(j)}\mathbf{H}^{-1}\bm{\Delta}\bH^{(j)}(\mathbf{H}^{(j)})^{-1}.
\label{eq:prof_Conv_2}
\end{align}
We can then use the properties of the norm (matrix norm inequality and triangular inequality) to obtain the upperbound
\begin{align}
&\|\left(\mathbf{H}^{(j)}\right)^{-1}\bm{e}_i-\mathbf{H}^{-1}\bm{e}_i\|\leq \sqrt{P^{-1}}\|\mathbf{H}^{-1}\|^2_{\Fro}\|\bm{\Delta}\bH^{(j)}\|_{\Fro}\notag\\
& \qquad\qquad\qquad+P^{-1}\|(\mathbf{H}^{(j)})^{-1}\|_{\Fro}\|\mathbf{H}^{-1}\|_{\Fro}^2 \|\bm{\Delta}\bH^{(j)}\|_{\Fro}^2
%&\doteq \sqrt{P^{-1}}\|\mathbf{H}^{-1}\|^2_{\Fro}\|\bm{\Delta}\bH^{(j)}\|_{\Fro}.
\label{eq:prof_Conv_3}
\end{align} 
It follows then directly from the norm properties that
\begin{align}
\left|\|\left(\mathbf{H}^{(j)}\right)^{-1}\!\!\!\!\!\!\bm{e}_i\|-\|\mathbf{H}^{-1}\bm{e}_i\|\right|&\leq \|\left(\mathbf{H}^{(j)}\right)^{-1}\!\!\!\!\!\!\bm{e}_i-\mathbf{H}^{-1}\bm{e}_i\|
\label{eq:prof_Conv_4}
\end{align} 
We can then use this result to write
\begin{align}
&\left\|\frac{\left(\mathbf{H}^{(j)}\right)^{-1}\!\!\!\!\!\!\bm{e}_i}{\|\left(\mathbf{H}^{(j)}\right)^{-1}\!\!\!\!\!\!\bm{e}_i\|}-\frac{\mathbf{H}^{-1}\bm{e}_i}{\|\mathbf{H}^{-1}\bm{e}_i\|}\right\|\\
&=
\|\left(\mathbf{H}^{(j)}\right)^{-1}\!\!\!\!\!\!\bm{e}_i\|\left\|\left(\mathbf{H}^{(j)}\right)^{-1}\!\!\!\!\!\!\bm{e}_i-\frac{ \|\left(\mathbf{H}^{(j)}\right)^{-1}\!\!\!\!\!\!\bm{e}_i\|}{\|\mathbf{H}^{-1}\bm{e}_i\|}\mathbf{H}^{-1}\bm{e}_i\right\|\\
%&=\|\left(\mathbf{H}^{(j)}\right)^{-1}\!\!\!\!\!\!\bm{e}_i\|\left\|\left(\mathbf{H}^{(j)}\right)^{-1}\!\!\!\!\!\!\bm{e}_i-\mathbf{H}^{-1}\bm{e}_i+\frac{ \|\mathbf{H}^{-1}\bm{e}_i\|-\|\left(\mathbf{H}^{(j)}\right)^{-1}\!\!\!\!\!\!\bm{e}_i\|}{\|\mathbf{H}^{-1}\bm{e}_i\|}\mathbf{H}^{-1}\bm{e}_i\right\|\\
&\leq\|\left(\mathbf{H}^{(j)}\right)^{-1}\!\!\!\!\!\!\bm{e}_i\|\left\|\left(\mathbf{H}^{(j)}\right)^{-1}\!\!\!\!\!\!\bm{e}_i-\mathbf{H}^{-1}\bm{e}_i\right\|\notag\\
&\qquad\qquad\qquad\qquad+\left| \|\mathbf{H}^{-1}\bm{e}_i\|-\|\left(\mathbf{H}^{(j)}\right)^{-1}\!\!\!\!\!\!\bm{e}_i\|\right|\\
&\leq\LB\|\left(\mathbf{H}^{(j)}\right)^{-1}\!\!\!\!\!\!\bm{e}_i\|+1\RB\left\|\left(\mathbf{H}^{(j)}\right)^{-1}\!\!\!\!\!\!\bm{e}_i-\mathbf{H}^{-1}\bm{e}_i\right\|.
\label{eq:prof_Conv_5}
\end{align} 
Using \eqref{eq:prof_Conv_3}, inside \eqref{eq:prof_Conv_5} yields 
\begin{equation}
\begin{aligned}
&\left\|\frac{\left(\mathbf{H}^{(j)}\right)^{-1}\bm{e}_i}{\|\left(\mathbf{H}^{(j)}\right)^{-1}\bm{e}_i\|}-\frac{\mathbf{H}^{-1}\bm{e}_i}{\|\mathbf{H}^{-1}\bm{e}_i\|}\right\|\\
&\leq \sqrt{P^{-1}}\|\mathbf{H}^{-1}\|^2_{\Fro}(\|\mathbf{H}^{-1}\|^2_{\Fro}+1)\|\bm{\Delta}\bH^{(j)}\|_{\Fro}\\
&+P^{-1}\|(\mathbf{H}^{(j)})^{-1}\|_{\Fro}\|\mathbf{H}^{-1}\|_{\Fro}^2 (\|\mathbf{H}^{-1}\|^2_{\Fro}\!+\!1)\|\bm{\Delta}\bH^{(j)}\|_{\Fro}^2.
%&\doteq \sqrt{P^{-1}}\|\mathbf{H}^{-1}\|^2_{\Fro}\|\bm{\Delta}\bH^{(j)}\|_{\Fro}.
\label{eq:prof_Conv_6}
\end{aligned}
\end{equation}   
Because of our preliminary remark in Appendix~\ref{app:preliminary}, we consider that all the expectations exist and are finite. Taking the square and the expectation, we obtain then
\begin{equation}
\E\LSB \left\|\sqrt{P}\frac{\left(\mathbf{H}^{(j)}\right)^{-1}\bm{e}_i}{\|\left(\mathbf{H}^{(j)}\right)^{-1}\bm{e}_i\|}-\sqrt{P}\frac{\mathbf{H}^{-1}\bm{e}_i}{\|\mathbf{H}^{-1}\bm{e}_i\|}\right\|\RSB\dotleq P^{0}.
\label{eq:prof_Conv_7}
\end{equation} 
Applying Proposition~\ref{prop_sufficient} concludes the proof.
\end{proof}
 
%%%%%%%%%%%%%%%%%%%%%%%%%%%%%%%%%%%%%%%%%%%%%%%%%%%%%%%%%%%%%%%%%%
%%%%%%%%%%%%%%%%%%%%%%%%%%%%%%%%%%%%%%%%%%%%%%%%%%%%%%%%%%%%%%%%%% 
\section{Proof of Theorem~\ref{thm}}\label{app:proof_thm}
%%%%%%%%%%%%%%%%%%%%%%%%%%%%%%%%%%%%%%%%%%%%%%%%%%%%%%%%%%%%%%%%%%
%%%%%%%%%%%%%%%%%%%%%%%%%%%%%%%%%%%%%%%%%%%%%%%%%%%%%%%%%%%%%%%%%% 

\begin{proof}
Let us focus now without loss of generality on the CSIT allocation at TX~$j$. Following the sufficient condition in Proposition~\ref{prop_sufficient}, the maximal DoF is achieved at every RX if
\begin{align}
\E\LSB \left |\frac{\bm{e}_j^{\He}\mathbf{H}^{-1}\bm{e}_i}{\|\mathbf{H}^{-1}\bm{e}_i\|}-\frac{\bm{e}_j^{\He}\LB\mathbf{H}^{(j)}\RB^{-1}\bm{e}_i}{\|\LB\mathbf{H}^{(j)}\RB^{-1}\bm{e}_i\|}\right|^2\RSB&\dotleq P^{-1},\qquad\forall i.
\label{eq:thm_proof_1} 
\end{align} 
Hence, we will in the following show that the CSIT allocation given in Theorem~\ref{thm} ensures that \eqref{eq:thm_proof_1} is satisfied. 

Following a similar calculation as in \eqref{eq:prof_Conv_5} in Appendix~\ref{app:proof_conv}, we obtain the relation given in \eqref{eq:thm_proof_2} at the top of next page.
\begin{figure*}
\begin{align}
 \left |\frac{\bm{e}_j^{\He}\mathbf{H}^{-1}\bm{e}_i}{\|\mathbf{H}^{-1}\bm{e}_i\|}-\frac{\bm{e}_j^{\He}\LB\mathbf{H}^{(j)}\RB^{-1}\bm{e}_i}{\|\LB\mathbf{H}^{(j)}\RB^{-1}\bm{e}_i\|}\right|&\leq 
\| \mathbf{H}^{-1} \bm{e}_i\| \left|\bm{e}_j^{\He}\mathbf{H}^{-1}\bm{e}_i-\bm{e}_j^{\He}(\mathbf{H}^{(j)})^{-1}\bm{e}_i\right|+|\bm{e}_j^{\He}(\mathbf{H}^{(j)})^{-1}\bm{e}_i| \left|\|\mathbf{H}^{-1}\bm{e}_i\|-\|(\mathbf{H}^{(j)})^{-1}\bm{e}_i\|\right|
\label{eq:thm_proof_2} 
\end{align} 
\end{figure*}
We will bound the first term of \eqref{eq:thm_proof_2} in Subsection~\ref{app:proof_thm:1} and the second term in Subsection~\ref{app:proof_thm:2}. 

Our calculations rely on the following series expansion for the channel inverse. It can be seen that the outer-diagonal of $\bH$ is at least equal to~$P^{-\min_{i\neq j}\Gamma_{i,j}}$. Upon defining the matrix~$\mathbf{D}\triangleq \diag(\mathbf{H})$, this means that we have for $P$ large enough 
\begin{equation}
\|\I_K-\bD^{-1}\bH\|_{\Fro}<1.
\label{eq:thm_proof_3} 
\end{equation}
It follows that we can define the geometric sum of the matrices as
\begin{equation}
\bS\triangleq \sum_{n=0}^{\infty}(\I_K-\bD^{-1}\bH)^n.
\label{eq:thm_proof_4} 
\end{equation}
Studying the partial sums of $\bS$, the following well known result can be shown
\begin{equation}
(\bD^{-1}\bH)^{-1}=\bS.
\label{eq:thm_proof_5} 
\end{equation}
This means that the elements of the channel inverse can be written as
\begin{align}
\bm{e}_j^{\He}\mathbf{H}^{-1}\bm{e}_i&=\bm{e}_j^{\He}\sum_{n=0}^\infty (\mathbf{D}^{-1}(\mathbf{D}-\mathbf{H}))^n\mathbf{D}^{-1}\bm{e}_i,\qquad \forall i\\
&=\sum_{n=0}^\infty  C^{j,i}_n
\label{eq:thm_proof_6} 
\end{align}
where we have defined
\begin{equation}
C^{j,i}_n \triangleq \bm{e}_j^{\He}(\mathbf{D}^{-1}(\mathbf{D}-\mathbf{H}))^n\mathbf{D}^{-1}\bm{e}_i,\qquad \forall i, n.
\label{eq:thm_proof_7} 
\end{equation}
Writing explicitly the matrix products in \eqref{eq:thm_proof_7}, it follows from basic algebra that
\begin{align}
\E[|C^{j,i}_n|^2]&\dotleq \bigg(P^ {-\min_{i\neq j}\Gamma_{i,j}} \bigg)^n.
\label{eq:thm_proof_8} 
\end{align} 
Hence, the $C_{ji}^n$ with $n$ such that 
\begin{equation}
n \min_{i\neq j}\Gamma_{i,j} >1
\label{eq:thm_proof_9} 
\end{equation}
can be neglected without any impact over the DoF. This translates to truncating the infinite summation to a finite summation up to~$n_{0}$ defined as
\begin{equation}
n_0\triangleq \bigg\lceil \frac{1}{\min_{i\neq j}\Gamma_{i,j}}\bigg\rceil.
\label{eq:thm_proof_10} 
\end{equation}

%%%%%%%%%%%%%%%%%%%%%%%%%%%%%%%%%%%%%%%%%%%%%%%%%%%%%%%%%%%%%%%%%% 
\subsection{Analysis of $\left|\bm{e}_j^{\He}\mathbf{H}^{-1}\bm{e}_i-\bm{e}_j^{\He}(\mathbf{H}^{(j)})^{-1}\bm{e}_i\right|$}\label{app:proof_thm:1}
%%%%%%%%%%%%%%%%%%%%%%%%%%%%%%%%%%%%%%%%%%%%%%%%%%%%%%%%%%%%%%%%%% 

We focus first on the first term, and since the norm~$\| \mathbf{H}^{-1} \bm{e}_i\|$ is equivalent to $H_{ii}^{-1}$ as the SNR~$P$ increases, we omit it in the following calculation for the sake of clarity.

We further introduce $\forall i,n$,
\begin{align}
\mathbf{D}^{(j)}&\triangleq \diag(\mathbf{H}^{(j)}),\\
C^{j,i,(j)}_n&\triangleq\bm{e}_j^{\He}((\mathbf{D}^{(j)})^{-1}(\mathbf{D}^{(j)}-\mathbf{H}^{(j)}))^n(\mathbf{D}^{(j)})^{-1}\bm{e}_i. 
\label{eq:thm_proof_11} 
\end{align}
We can then write
\begin{align}
&\E[|\bm{e}_j^{\He}\mathbf{H}^{-1}\bm{e}_i-\bm{e}_j^{\He}\LB\mathbf{H}^{(j)}\RB^{-1}\!\!\!\!\!\!\bm{e}_i|^2]\notag
\\&\doteq 
\E\bigg[\bigg|\sum_{n=1}^{n_{0}} C^{j,i}_n-C^{j,i,(j)}_n\bigg|^2\bigg]\\
&\leq \E\bigg[\LB\sum_{n=1}^{n_{0}} |C^{j,i}_n-C^{j,i,(j)}_n|\RB^2\bigg]\\
&\dotleq \sum_{n=1}^{n_{0}} \E[|C^{j,i}_n-C^{j,i,(j)}_n|^2]
\label{eq:thm_proof_12} 
\end{align}
where we have used iteratively that $(a+b)^2\leq 2(a^2+b^2), \forall a,b \in \mathbb{R}^2$ to obtain the last inequality (and the multiplicative constants could be removed because of the exponential inequality). We now look for a CSIT allocation~$\bB^{(j)}$ ensuring that  
\begin{align}
 \E[|C^{j,i}_n-C^{j,i,(j)}_n|^2]\dotleq P^{-1},\qquad \forall i,n.
\label{eq:thm_proof_13} 
\end{align}
As a starting point, we consider the first coefficients.
\begin{align}
\E[|C^{j,i}_0-C^{j,i,{(j)}}_0|^2]&=\E\bigg[\bigg|\frac{\bm{e}_j^{\He}\bm{e}_i}{H_{i,i}}-\frac{\bm{e}_j^{\He}\bm{e}_i}{H_{i,i}^{(j)}}\bigg|^2\bigg]\\
%&=\E\bigg[\bigg|\frac{\sigma_{ii}^{(j)}\Delta{H}^{(j)}_{ii}}{H_{ii}H_{ii}^{(j)}}\bigg|^2\bigg]\delta_{ji}\\
&\doteq 2^{-B_{i,i}^{(j)}} \delta_{j,i}.
\label{eq:thm_proof_14} 
\end{align}
Thus, we have
\begin{align}
B_{j,j}^{(j)}\geq\lceil \log_2(P)\rceil.
\label{eq:thm_proof_15} 
\end{align}
This ensures to fulfill \eqref{eq:thm_proof_13}  for $n=0$. The error done over $H_{j,j}$ becomes then negligible in terms of DoF (i.e., in terms of exponential equality). For $n\geq 1$, it holds that~$C^{j,j}_n=0$ such that we assume that~$i\neq j$ in the following,
\begin{align}
\E[|C^{j,i}_1-C^{j,i,{(j)}}_1|^2]&\!=\!\E\bigg[\bigg|\frac{ \{\mathbf{D}-\mathbf{H}\}_{j,i}}{H_{j,j}H_{i,i}}-\frac{ \{\mathbf{D}^{(j)}-\mathbf{H}^{(j)}\}_{j,i}}{H^{(j)}_{j,j}H^{(j)}_{i,i}}\bigg|^2\bigg]\\
&\doteq \E\bigg[\bigg|\frac{ H_{i,i}^{(j)}\tilde{H}_{j,i}-H_{i,i}\tilde{H}_{j,i}^{(j)}}{H_{j,j}H_{i,i}H^{(j)}_{i,i}}\bigg|^2\bigg] P^{-\Gamma_{j,i}}\\
%&\doteq  ((\sigma_{ii}^{(j)})^2+(\sigma_{ji}^{(j)})^2)P^{-\{\bm{\Gamma}\}_{i,j}}\\
&\doteq  (2^{-B_{i,i}^{(j)}}+2^{-B_{j,i}^{(j)}})P^{-\Gamma_{j,i}}
\label{eq:thm_proof_16} 
\end{align}
Setting 
\begin{align}
B_{i,i}^{(j)}&\geq\lceil [1-\Gamma_{j,i}]^{+}\log_2(P)\rceil ,\qquad \forall i\neq j\label{eq:General_proof_11_1} \\
B_{j,i}^{(j)}&\geq\lceil[1-\Gamma_{j,i}]^{+}\log_2(P)\rceil ,\qquad \forall i\neq j
\label{eq:thm_proof_17} 
\end{align}
ensures to fulfill \eqref{eq:thm_proof_13} for all streams~$i$ for $n=1$. Going further, we consider then~$C^{ji}_2$,
\begin{align}
&\E[|C^{j,i}_2-C^{j,i,{(j)}}_2|^2]\\
&=\E\bigg[ \bigg|\bm{e}_j^{\He}(\mathbf{D}^{-1}(\mathbf{D}-\mathbf{H}))^2\mathbf{D}^{-1}\bm{e}_i\notag\\
&\qquad- \bm{e}_j^{\He}((\mathbf{D}^{(j)})^{-1}(\mathbf{D}^{(j)}-\mathbf{H}^{(j)}))^2(\mathbf{D}^{(j)})^{-1}\bm{e}_i\bigg|^2\bigg]\\
%&=\E\bigg[ \bigg|\sum_{k=1}^K\bm{e}_j^{\He}\mathbf{D}^{-2}(\mathbf{D}-\mathbf{H})\bm{e}_k\bm{e}_k^{\He}(\mathbf{D}-\mathbf{H})\mathbf{D}^{-1}\bm{e}_i\notag\\
%&-\bm{e}_j^{\He}(\mathbf{D}^{(j)})^{-2}(\mathbf{D}^{(j)}\!-\!\mathbf{H}^{(j)})\bm{e}_k\bm{e}_k^{\He}(\mathbf{D}^{(j)}\!-\!\mathbf{H}^{(j)})(\mathbf{D}^{(j)})^{-1}\bm{e}_i\bigg|^2\!\bigg]\\
&=\E\bigg[ \bigg|\sum_{k=1, k\neq i, k\neq j}^K \frac{1}{H_{j,j}^{2}H_{i,i}}\tilde{H}_{j,k}\tilde{H}_{k,i}\notag\\
&~~~~~~~~~~~~~~-\frac{1}{(H_{j,j}^{(j)})^{2}H_{i,i}^{(j)}}\tilde{H}_{j,k}^{(j)}\tilde{H}_{k,i}^{(j)}\bigg|^2\bigg]P^{-(\Gamma_{j,k}+\Gamma_{k,i})}\\
%&\dotleq\E[ |\sum_{k=1, k\neq i, k\neq j}^K  \tilde{H}_{jk}\tilde{H}_{ki}-\tilde{H}_{jk}^{(j)}\tilde{H}_{ki}^{(j)}|^2](\mu^2)^{\dist(j,k)+\dist(k,i)}\\
&\dotleq\E\bigg[ \sum_{k=1, k\neq i, k\neq j}^K  |\tilde{H}_{j,k}\tilde{H}_{k,i}-\tilde{H}_{j,k}^{(j)}\tilde{H}_{k,i}^{(j)}\bigg|^2\bigg]P^{-(\Gamma_{j,k}+\Gamma_{k,i})}\\
&\doteq \sum_{k=1, k\neq i, k\neq j}^K  (2^{-B_{j,k}^{(j)}}-2^{B_{k,i}^{(j)}}) P^{-(\Gamma_{j,k}+\Gamma_{k,i})}\\
&\doteq \sum_{k=1, k\neq i, k\neq j}^K 2^{-B_{k,i}^{(j)}} P^{-(\Gamma_{j,k}+\Gamma_{k,i})}
\label{eq:thm_proof_18} 
\end{align}
where we could remove~$2^{-B_{j,k}^{(j)}}$ in the last exponential inequality because of \eqref{eq:thm_proof_17}. Setting
\begin{align}
B_{k,i}^{(j)}&\geq \lceil [1-(\Gamma_{j,k}+\Gamma_{k,i})]^{+}\log_2(P)\rceil , \forall k,i
\label{eq:thm_proof_19} 
\end{align}
allows to fulfill \eqref{eq:thm_proof_13} for $n=2$. Going to arbitrary value of~$n$, we obtain equation \eqref{eq:thm_proof_20} at the top of next page.
\begin{figure*}
\begin{align}
&\E[|C^{j,i}_n-C^{j,i,{(j)}}_n|^2]\nonumber\\
&=\E\bigg[\bigg|\sum_{k_1\neq j}^K \!\sum_{k_2\neq k_1}^K \!\ldots\!\! \sum_{\substack{k_{n-1}\neq k_{n-2}\\k_{n-1}\neq i}}^K \!\!\bigg(\frac{\tilde{H}_{j,k_1}\tilde{H}_{k_1,k_2}\ldots,\tilde{H}_{k_{n-1,i}}}{H_{j,j}^{n} H_{i,i}}-\frac{\tilde{H}^{(j)}_{j,k_1}\tilde{H}^{(j)}_{k_1,k_2}\ldots,\tilde{H}^{(j)}_{k_{n-1,i}}}{(H^{(j)}_{j,j})^{n} H_{i,i}^{(j)}} \bigg)\bigg|^2\bigg]
P^{-(\Gamma_{j,k_1}+\Gamma_{k_1,k_2}+\ldots+\Gamma_{k_{n-1},i})}.
\label{eq:thm_proof_20} 
\end{align}
\end{figure*}  
It is clear from \eqref{eq:thm_proof_20}  that when considering the $i$th stream, the coefficient $\tilde{H}_{k,\ell}$ appears weighted with $P^{-1}$ at a coefficient at least equal to $\Gamma_{k\rightarrow j}+\Gamma_{k,\ell}+\Gamma_{i\rightarrow \ell}$. Yet, this condition has to be fulfilled for every stream, hence every value of $i$. Since~$\Gamma_{i\rightarrow \ell}\geq 0$, taking $i=\ell$ gives the most tight constraint. In total, this gives the CSIT allocation
\begin{align}
B_{k,\ell}^{(j)}&\geq \lceil [1-(\Gamma_{k\rightarrow j}+\Gamma_{k,\ell})]^{+}\log_2(P)\rceil ,\qquad  \forall k,\ell.
\label{eq:thm_proof_21} 
\end{align}
This CSIT allocation can be seen to fulfill \eqref{eq:thm_proof_13} for all $i$ and all $n$.
%%%%%%%%%%%%%%%%%%%%%%%%%%%%%%%%%%%%%%%%%%%%%%%%%%%%%%%%%%%%%%%%%% 
\subsection{Analysis of $|\bm{e}_j^{\He}(\mathbf{H}^{(j)})^{-1}\bm{e}_i| \left|\|\mathbf{H}^{-1}\bm{e}_i\|-\|(\mathbf{H}^{(j)})^{-1}\bm{e}_i\|\right|$}\label{app:proof_thm:2}
%%%%%%%%%%%%%%%%%%%%%%%%%%%%%%%%%%%%%%%%%%%%%%%%%%%%%%%%%%%%%%%%%% 
From the norm properties, it holds that
\begin{align}
&\left|\|\mathbf{H}^{-1}\bm{e}_i\|-\|(\mathbf{H}^{(j)})^{-1}\bm{e}_i\|\right|^2\notag\\
&\leq \left \|\mathbf{H}^{-1}\bm{e}_i - (\mathbf{H}^{(j)})^{-1}\bm{e}_i\right\|^2\\
&= \sum_{k=1}^K \left |\bm{e}_k^{\He} \mathbf{H}^{-1}\bm{e}_i -\bm{e}_k^{\He} (\mathbf{H}^{(j)})^{-1}\bm{e}_i\right|^2
\label{eq:thm_proof_22} 
\end{align}
Hence, we can study the CSIT requirements ensuring that
\begin{align}
 \E[|C^{k,i}_n-C^{k,i,(j)}_n|^2]\dotleq P^{-1},\qquad \forall k,n.
\label{eq:thm_proof_23} 
\end{align}
It can be seen that the difference between \eqref{eq:thm_proof_23} and \eqref{eq:thm_proof_13} comes from the fact that \eqref{eq:thm_proof_13} had to be fulfilled for every stream~$i$, whereas \eqref{eq:thm_proof_23} has to be fulfilled for every $k$, with the stream index fixed. Proceeding similarly as in the first part, we obtain the following condition
\begin{align}
B_{k,\ell}^{(j)}\geq 1-\Gamma_{k,\ell}-\Gamma_{i\rightarrow \ell},\qquad \forall k,\ell.
\label{eq:thm_proof_24} 
\end{align}
However, we have also shown in the previous calculations that
\begin{align}
|\bm{e}_j^{\He}(\mathbf{H}^{(j)})^{-1}\bm{e}_i|^2\dotleq P^{-\Gamma_{j,i}}.
\label{eq:thm_proof_25} 
\end{align}
Taking this further attenuation into account and the fact that the CSIT requirements have to be fulfilled for every stream~$i$, we obtain the CSIT allocation
\begin{align}
B_{k,\ell}^{(j)}\geq 1-\Gamma_{k,\ell}-\min_{i}\left (\Gamma_{i\rightarrow \ell}+\Gamma_{j,i}\right).
\label{eq:thm_proof_26} 
\end{align}
Putting together \eqref{eq:thm_proof_21} and \eqref{eq:thm_proof_26} gives the CSIT allocation in Theorem~\ref{thm}.
\end{proof}
%\bibliographystyle{IEEEtran}
%\bibliography{./../Literatur}

  \begin{IEEEbiography}[{\includegraphics[scale=0.5]{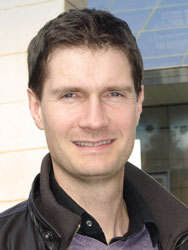}}]{David Gesbert}
	 (IEEE Fellow) is Professor and Head of the Mobile Communications Department, EURECOM, France, where he also heads the Communications Theory Group. He obtained the Ph.D degree from Ecole Nationale Superieure des Telecommunications, France, in 1997. From 1997 to 1999 he has been with the Information Systems Laboratory, Stanford University. In 1999, he was a founding engineer of Iospan Wireless Inc, San Jose, Ca.,a startup company pioneering MIMO-OFDM (now Intel). Between 2001 and 2003 he has been with the Department of Informatics, University of Oslo as an adjunct professor. D. Gesbert has published over 200 papers and several patents all in the area of signal processing, communications, and wireless networks.

 D. Gesbert was a co-editor of several special issues on wireless networks and communications theory, for JSAC (2003, 2007, 2009), EURASIP Journal on Applied Signal Processing (2004, 2007), Wireless Communications Magazine (2006). He served on the IEEE Signal Processing for Communications Technical Committee, 2003-2008.  He was an associate editor for IEEE Transactions on Wireless Communications and the EURASIP Journal on Wireless Communications and Networking. He authored or co-authored papers winning the 2012 SPS Signal Processing Magazine Best Paper Award, 2004 IEEE Best Tutorial Paper Award (Communications Society), 2005 Young Author Best Paper Award for Signal Proc. Society journals, and paper awards at conferences 2011 IEEE SPAWC, 2004 ACM MSWiM workshop. He co-authored the book Space time wireless communications: From parameter estimation to MIMO systems, Cambridge Press, 2006.  In 2013, he was a General Chair for the IEEE Communications Theory Workshop, and a Technical Program Chair for the Communications Theory Symposium of ICC2013. He is a Technical Program Chair for IEEE ICC 2017, to be held in Paris. 
 \end{IEEEbiography}
 \smallskip{} 
  \begin{IEEEbiography}[{\includegraphics[scale=0.4]{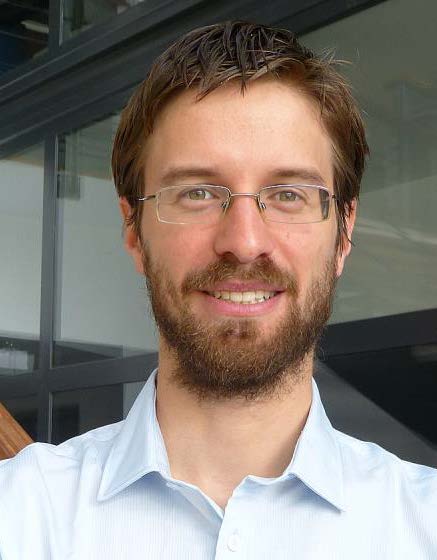}}]{Paul de Kerret} (IEEE Student Member) graduated in 2009 from Ecole Nationale Superieure des Telecommunications de Bretagne, France and obtained a diploma degree in electrical engineering from Munich University of Technology (TUM), Germany. He also earned a four year degree in mathematics at the Universite de Bretagne Occidentale, France in 2008. From January 2010 to september 2010, he has been a research assistant at the Institute for Theoretical Information Technology, RWTH Aachen University, Germany. In december 2013, he obtained a Ph.D. degree in the Mobile Communications Department at EURECOM, France, under the supervision of David Gesbert. He is the first author of several journal papers in prestigious journals and a magazine. With David Gesbert, he will give a tutorial in ICASSP 2014 on the challenges behind the cooperation of transmitters in wireless networks. He is interested in communication theory, information theory, and game theory.
 \end{IEEEbiography}
\end{document}